\documentclass[a4paper,twocolumn,accepted=2019-11-05]{quantumarticle}
\pdfoutput=1

\usepackage[utf8]{inputenc}
\usepackage[english]{babel}
\usepackage[T1]{fontenc}

\usepackage{graphicx}
\graphicspath{ {figures_all/} }

\usepackage{amsmath}
\usepackage{amssymb}
\usepackage{xypic}
\usepackage{amsthm}

\allowdisplaybreaks

\usepackage{hyperref}

\hypersetup{
 bookmarks=true,		
 unicode=false,			
 pdftoolbar=true,		
 pdfmenubar=true,		
 pdffitwindow=false,		
 pdfstartview={FitH},		
 pdftitle={k-stretchability of entanglement, and the duality of k-separability and k-producibility},    
 pdfauthor={Szil{\'a}rd Szalay},	
 pdfsubject={research article on multipartite entanglement},	
 pdfcreator={pdflatex},		
 pdfproducer={vim},		
 pdfkeywords={quantum} {entanglement} {multipartite} {mixed states}, 
 pdfnewwindow=true,		
 colorlinks=true,		
 linktoc=page,			
 linkcolor=blue,		
 citecolor=blue,		
 filecolor=blue,		
 urlcolor=blue			
}

\usepackage[numbers,sort&compress]{natbib}
\newcommand{\field}[1]{\mathbb{#1}}
\newcommand{\ve}[1]{\mathbf{#1}}
\newcommand{\vs}[1]{\boldsymbol{#1}}
\newcommand{\vvs}[1]{\underline{\boldsymbol{#1}}}

\newcommand{\cket}[1]{\vert #1 \rangle}
\newcommand{\bra}[1]{\langle #1 \vert}

\providecommand{\abs}[1]{{\lvert#1\rvert}}

\newcommand{\cmpl}[1]{\overline{#1}}

\newcommand{\DscGrp}[1]{\mathrm{#1}}

\newcommand{\pconj}{\dagger}

\DeclareMathOperator{\Tr}{Tr}

\DeclareMathOperator{\Conv}{Conv}

\DeclareMathOperator{\Span}{Span}

\DeclareMathOperator{\downset}{\downarrow}
\DeclareMathOperator{\upset}{\uparrow}

\newcommand{\finereq}{\preceq}

\newcommand{\coarsereq}{\succeq}

\newcommand{\equals}{{\;\;=\;\;}}
\newcommand{\lequals}{{\;\;\leq\;\;}}
\newcommand{\equalsref}[1]{\overset{\eqref{#1}}{\equals}}
\newcommand{\lequalsref}[1]{\overset{\eqref{#1}}{\lequals}}

\newcommand{\inliff}{\,\Leftrightarrow\,}

\newcommand{\inlthen}{\,\Rightarrow\,}

\newcommand{\dspiff}{\;\;\Longleftrightarrow\;\;}

\newcommand{\dspthen}{\;\;\Longrightarrow\;\;}

\newcommand{\dspdef}{\;\;\overset{\text{def.}}{\Longleftrightarrow}\;\;}

\newcommand{\dispt}[1]{\;\;\text{#1}\;\;}

\newcommand{\cnvroof}[1]{{#1}^{\cup}}

\newcommand{\isom}{\cong}

\providecommand{\abs}[1]{\lvert#1\rvert}

\newcommand{\set}[1]{\{ #1 \}}
\newcommand{\bigset}[1]{\bigl\{ #1 \bigr\}}
\newcommand{\Bigset}[1]{\Bigl\{ #1 \Bigr\}}

\newcommand{\sset}[2]{\set{ #1 \;\vert\; #2 }}
\newcommand{\bigsset}[2]{\bigset{ #1 \;\big\vert\; #2 }}
\newcommand{\Bigsset}[2]{\Bigset{ #1 \;\Big\vert\; #2 }}

\newcommand{\mset}[1]{\set{#1}}
\newcommand{\bigmset}[1]{\bigset{#1}}

\newcommand{\smset}[2]{\mset{ #1 \;\vert\; #2 }}
\newcommand{\bigsmset}[2]{\bigmset{ #1 \;\big\vert\; #2 }}

\newcommand{\pinv}[1]{\hat{#1}}

\newtheorem{thm}{Theorem}
\newtheorem{lem}[thm]{Lemma}
\newtheorem{cor}[thm]{Corollary}

\begin{document}
\title{k-stretchability of entanglement,
and the duality of k-separability and k-producibility}
\author{{Sz}il{\'a}rd {Sz}alay}
\email{szalay.szilard@wigner.mta.hu}
\affiliation{
Strongly Correlated Systems ``Lend{\"u}let'' Research Group,
Wigner Research Centre for Physics, 
29-33, Konkoly-Thege Mikl{\'o}s str., H-1121 Budapest, Hungary}
\orcid{0000-0002-1214-6811}

\begin{abstract}
The notions of 
$k$-separability and $k$-producibility
are useful and expressive tools for the characterization of
entanglement in multipartite quantum systems,
when a more detailed analysis would be infeasible or simply needless.
In this work we reveal a partial duality between them,
which is valid also for their correlation counterparts.
This duality can be seen from a much wider perspective,
when we consider the entanglement and correlation properties which are invariant under the permutations of the subsystems.
These properties are labeled by Young diagrams, which we endow with a refinement-like partial order,
to build up their classification scheme.
This general treatment reveals a new property, which we call $k$-stretchability,
being sensitive in a balanced way to 
both the maximal size of correlated (or entangled) subsystems
 and the minimal number of subsystems uncorrelated with (or separable from) one another.
\end{abstract}


\maketitle{}

\tableofcontents{}

\section{Introduction}
\label{sec:intro}


The investigation of the \emph{correlations} among the parts of a composite physical system is
an essential tool in statistical physics.
If the system is described by quantum mechanics, 
then nonclassical forms of correlations arise, 
the most notable is \emph{entanglement} \cite{Schrodinger-1935a,Schrodinger-1935b,Horodecki-2009,Modi-2010}.
It is the main resource of \emph{quantum information theory} \cite{Nielsen-2000,Petz-2008,Wilde-2013},
and its nonclassical properties, playing important role also in many-body physics \cite{Coffman-2000,Koashi-2004,Eisert-2010},
make it influential in the behavior and characterization of \emph{strongly correlated systems}
\cite{Amico-2008,Legeza-2004,Szalay-2015a,Szalay-2017}.

The correlation and entanglement between \emph{two parts} of a system is relatively well-understood \cite{Werner-1989,Bennett-1996a,Bennett-1996b}.
At least for pure states, the \emph{convertibility} and the \emph{classification} with respect to (S)LOCC
((stochastic) local operations and classical communication \cite{Bennett-1996b,Werner-1989,Chitambar-2014,Dur-2000b})
shows a simple structure \cite{Nielsen-1999,Vidal-1999},
and there are basically unique correlation and entanglement \emph{measures} \cite{Vidal-2000,Horodecki-2001,Plenio-2007}.
The \emph{multipartite} case is much more complicated \cite{Dur-2000b,Verstraete-2002,Luque-2003,Eltschka-2014,Yamasaki-2018}.
Even for pure states,
the nonexistence of a maximally entangled reference state \cite{Hebenstreit-2016,Spee-2016} 
and the involved nature of state-transformations in general \cite{Gour-2017,Sauerwein-2018}
seem to make the (S)LOCC-based \emph{classification} practically unaccomplishable,
and the standard (S)LOCC paradigm less enlightening so less expressive.
Taking into account partial entanglement (partial separability)
\cite{Dur-1999,Dur-2000a,Acin-2001,Nagata-2002,Seevinck-2008,
Szalay-2011,Szalay-2012,Szalay-2015b},
or partial correlations \cite{Szalay-2017,Szalay-2018,Brandejs-2018} only,
leads to a combinatoric \cite{Davey-2002,Roman-2008,Stanley-2012}, discrete classification
(based on the \emph{lattice of set partitions} \cite{Roman-2008}),
endowed naturally with a well-behaving set of correlation and entanglement measures, 
characterizing the finite number of properties \cite{Szalay-2015b,Szalay-2017}.

Even the partial correlation and partial entanglement properties 
are getting too involved rapidly, with the increasing of the number of subsystems.
Singling out particularly expressive properties,
we consider the \emph{$k$-partitionability} and \emph{$k$-producibility} of correlation and of entanglement.
For $n$ subsystems, $k$ ranges from $1$ to $n$,
so the number of these properties scales linearly with the number of subsystems,
moreover, their structures are the simplest possible ones, chains.
(In the case of entanglement,
$k$-partitionability is called \emph{$k$-separability} \cite{Acin-2001,Seevinck-2008},
while $k$-producibility \cite{Seevinck-2001,Guhne-2005,Guhne-2006,Toth-2010} 
is also called \emph{entanglement depth} \cite{Sorensen-2001,Lucke-2014,Chen-2016}.
Here we use both concepts for correlation and also for entanglement,
this is why we use the naming
``$k$-partitionability of correlation''  and ``$k$-producibility of correlation'',
``$k$-partitionability of entanglement'' and ``$k$-producibility of entanglement'' \cite{Szalay-2015b,Szalay-2017,Szalay-2018}.)
These characterize the strength of two different (one-parameter-) aspects of multipartite correlation and entanglement:
those which cannot be restricted inside at least $k$ parts,
and
those which cannot be restricted inside parts of size at most $k$, respectively.
The concepts of $k$-partitionability and $k$-producibility of entanglement 
found application in spin chains \cite{Guhne-2005,Guhne-2006},
appeared in quantum nonlocality \cite{Curchod-2015}, 
leading to device independent certification of them \cite{Liang-2015,Lin-2019},
and were also demonstrated in quantum optical experiment \cite{Lu-2018}.
$k$-producibility also plays particularly important role in quantum metrology \cite{Toth-2014}.
$k$-producibly entangled states for larger $k$ lead to higher sensitivity,
so better precision in phase estimation,
which has been illustrated in experiments \cite{Hyllus-2012,Gessner-2018,Qin-2019},
and which also leads to $k$-producibility entanglement criteria \cite{Toth-2012}.

$k$-partitionability and $k$-producibility
are special cases of properties \emph{invariant under the permutation of the subsystems}.
The permutation invariant correlation properties are
based on the \emph{integer partitions} \cite{Andrews-1984,Stanley-2012}, also known (represented) as Young diagrams.
For the description of the structure of these,
we introduce a new order over the integer partitions, called refinement,
induced by the refinement order over set partitions,
used for the description of the structure of the partial correlation or entanglement properties \cite{Szalay-2015b}.
The structure of the permutation invariant properties contains the chains of $k$-partitionability and $k$-producibility,
and it is simpler than the structure describing all the properties. 
The number of these scales still rapidly
with the number of subsystems,
but slower than the number of all the partial correlation or entanglement properties.

The general treatment of the permutation invariant properties
reveals a partial duality between $k$-partitionability and $k$-producibility,
which is the manifestation of a duality on a deeper level, relating important properties of Young diagrams,
by which $k$-partitionability and $k$-producibility are formulated.

The general treatment of the permutation invariant properties
reveals also a particularly expressive new property, which we call \emph{$k$-stretchability},
leading to the definitions of ``$k$-stretchability of correlation'' and ``$k$-stretchability of entanglement''.
It combines the advantages 
of $k$-partitionability and $k$-producibility.
Namely, $k$-partitionability is about the number of subsystems uncorrelated with (or separable from) one another, and 
  not sensitive to the size   of correlated   (or entangled) subsystems \cite{Guhne-2005};
and     $k$-producibility    is about the size   of the largest correlated (or entangled) subsystem, and
  not sensitive to the number of subsystems uncorrelated with (or separable from) one another;
while $k$-stretchability is sensitive to both of these, in a balanced way.
The price to pay for this is that we have roughly twice as many $k$-stretchability properties, $k$ goes from $-(n-1)$ to $n-1$.
However, $k$-stretchability is linearly ordered with $k$, 
while the relations between $k$-partitionability and $k$-producibility are far more complicated.

The organization of this work is as follows.
In Section~\ref{sec:general}, we recall the structure of multipartite correlation and entanglement.
In Section~\ref{sec:perminv}, we construct the parallel structure for the permutation invariant case.
In Section~\ref{sec:altperminv}, we show another introduction of the permutation invariant classification, 
which can although be considered simpler, but less transparent.
In Section~\ref{sec:kpps}, we recall $k$-partitionability and $k$-producibility,
introduce $k$-stretchability, and show how these properties are related to each other.
In Section~\ref{sec:summ}, summary, remarks and open questions are listed.
In Appendix~\ref{app:structPI}, we work out the ``coarsening'' step and some other tools, used in the main text.
In Appendix~\ref{app:misc}, we present the proofs of some further propositions given in the main text.

In the course of the presentation,
the abstract mathematical structure of the correlation and entanglement properties 
(given in terms of lattice theoretic constructions)
is well-separated
from the concrete hierarchies of the state sets and measures.
This leads to a transparent construction,
which is easy to restrict to the permutation invariant case later.
For the convenience of the reader, we provide a short summary on the notations.
The abstract correlation and entanglement properties are labeled 
by Greek letters $\xi,\upsilon$.
On the three levels of the construction, these are typesetted 
as $\xi,\vs{\xi},\vvs{\xi}$ for the general case,
and as $\pinv{\xi},\pinv{\vs{\xi}},\pinv{\vvs{\xi}}$ for the permutation invariant case.
On the first two levels of the construction,
there are sets of quantum states of given correlation and entanglement properties,
typesetted as 
$\mathcal{D}_{\xi\text{-unc}}$,
$\mathcal{D}_{\xi\text{-sep}}$,
$\mathcal{D}_{\vs{\xi}\text{-unc}}$,
$\mathcal{D}_{\vs{\xi}\text{-sep}}$ for the general case, and
$\mathcal{D}_{\pinv{\xi}\text{-unc}}$,
$\mathcal{D}_{\pinv{\xi}\text{-sep}}$,
$\mathcal{D}_{\pinv{\vs{\xi}}\text{-unc}}$,
$\mathcal{D}_{\pinv{\vs{\xi}}\text{-sep}}$ for the permutation invariant case.
On the first two levels of the construction,
there are LO(CC)-monotonic measures of given correlation and entanglement properties,
typesetted as
$C_{\xi}$,
$E_{\xi}$,
$C_{\vs{\xi}}$,
$E_{\vs{\xi}}$ for the general case, and
$C_{\pinv{\xi}}$,
$E_{\pinv{\xi}}$,
$C_{\pinv{\vs{\xi}}}$,
$E_{\pinv{\vs{\xi}}}$ for the permutation invariant case.
On the third level of the construction,
there are classes (disjoint sets) of quantum states of given correlation and entanglement properties,
typesetted as 
$\mathcal{C}_{\vvs{\xi}\text{-unc}}$,
$\mathcal{C}_{\vvs{\xi}\text{-sep}}$ for the general case, and
$\mathcal{C}_{\pinv{\vvs{\xi}}\text{-unc}}$,
$\mathcal{C}_{\pinv{\vvs{\xi}}\text{-sep}}$ for the permutation invariant case.
The abstract
correlation and entanglement properties are ordered, which is denoted by $\preceq$ on the three levels of the construction
in both the general and the permutation invariant cases.
On the first two levels of the construction,
thanks to the monotonicity properties,
these manifest themselves as $\subseteq$ and $\geq$ for the state sets and for the measures, respectively.
On the third level of the construction, 
these manifest themselves as LO(CC) conversion results, for which we do not introduce notation.

\section{Multipartite correlation and entanglement}
\label{sec:general}

Here we recall the structure of the classification and quantification 
of multipartite correlation and entanglement \cite{Szalay-2015b,Szalay-2017,Szalay-2018}.
Our goal is to do this in the way sufficient to see how the
permutation invariant properties can be formulated parallel to this in the next section.

\subsection{Level 0: subsystems}
\label{sec:general.L0}

The classification scheme we present here is rather general.
The elementary and composite subsystems can be any discrete finite systems
possessing probabilistic description,
supposed that the joint systems can be represented by the use of \emph{tensor products,}
which is the basic tool in the constructions.
Such systems can be 
distinguishable quantum systems,
second quantized bosonic systems,
second quantized fermionic systems with fermion number parity superselection rule imposed,
or even classical systems, with significant simplification in the structure in the latter case \cite{Szalay-2018}.

Let $L$ be the set of the labels of $\abs{L}=n$ \emph{elementary subsystems.}
All the (possibly composite) \emph{subsystems} are then labeled by the subsets $X\subseteq L$,
the set of which, $P_0:=2^L$, naturally possesses a Boolean lattice structure 
with respect to the inclusion $\subseteq$, which is now denoted with $\finereq$.
For every elementary subsystem $i\in L$,
let the Hilbert space $\mathcal{H}_i$ be associated with it,
where $1<\dim\mathcal{H}_i<\infty$.
From these, 
for every subsystem $X\in P_0$,
the Hilbert space associated with it is $\mathcal{H}_X = \bigotimes_{i\in X}\mathcal{H}_i$. 
(For the trivial subsystem $X=\emptyset$, 
we have the one-dimensional Hilbert space $\mathcal{H}_\emptyset=\Span\set{\cket{\psi}}\isom\field{C}$
\cite{Szalay-2015b}.)
The \emph{states} of the subsystems $X\in P_0$
are given by density operators (positive semidefinite operators of trace $1$)
acting on $\mathcal{H}_X$,
the sets of those are denoted with $\mathcal{D}_X$.

The \emph{mixedness} of a state $\varrho_X\in\mathcal{D}_X$ of a subsystem $X$
can be characterized by the \emph{von Neumann entropy} \cite{Neumann-1927,Ohya-1993,Bengtsson-2006,Wilde-2013}
\begin{subequations}
\begin{equation}
\label{eq:S}
S(\varrho_X)=-\Tr\varrho_X\ln\varrho_X;
\end{equation}
and the \emph{distinguishability} of two states $\varrho_X,\sigma_X\in\mathcal{D}_X$ of subsystem $X$
can be characterized by the \emph{Umegaki relative entropy} \cite{Umegaki-1962,Lindblad-1973,Hiai-1991,Ohya-1993,Bengtsson-2006,Petz-2008,Wilde-2013} 
\begin{equation}
\label{eq:D}
D(\varrho_X\Vert\sigma_X)=\Tr\varrho_X(\ln\varrho_X-\ln\sigma_X).
\end{equation}
\end{subequations}
The von Neumann entropy is monotone increasing in bistochastic quantum channels,
and the relative entropy is monotone decreasing in quantum channels \cite{Lindblad-1973,Bengtsson-2006,Petz-2008}
and $S(\varrho_X)=0 \inliff \varrho=\cket{\psi}\bra{\psi}$,
$D(\varrho_X\Vert\sigma_X)=0 \inliff \varrho_X=\sigma_X$.

\subsection{Level~I: set partitions}
\label{sec:general.LI}

For handling the different possible splits of a composite system into subsystems,
we need to use the mathematical notion of (set) partition \cite{Roman-2008} of the system $L$.
The \emph{partitions} of $L$ are sets of subsystems,
\begin{subequations}
\begin{equation}
\xi=\set{X_1,X_2,\dots,X_{\abs{\xi}}}\in \Pi(L), 
\end{equation}
where the \emph{parts}, $X\in\xi$, are nonempty disjoint subsystems, for which $\bigcup_{X\in\xi}X=L$.
The set of the partitions of $L$ is denoted with
\begin{equation}
\label{eq:PI}
P_\text{I}:=\Pi(L).
\end{equation}
Its size is given by the \emph{Bell numbers} \cite{oeisA000110}, which are rapidly growing with $n$.
There is a natural partial order over the partitions,
which is called \emph{refinement} $\finereq$,
given as
\begin{equation}
\label{eq:poI}
\upsilon\finereq\xi \dspdef \forall Y\in\upsilon, \exists X\in\xi \dispt{s.t.} Y\subseteq X.
\end{equation}
\end{subequations}
(For illustrations, see Figure~\ref{fig:Ps23}.)
This grabs our natural intuition of comparing splits of systems.
Note that the refinement is a partial order only,
there are pairs of partitions which cannot be ordered,
for example, 
$12|3 \npreceq 13|2$ and
$12|3 \nsucceq 13|2$, see Figure~\ref{fig:Ps23}.
(We use a simplified notation for the partitions, for example, $12|3\equiv\set{\set{1,2},\set{3}}$,
where this does not cause confusion.)
Note that $P_\text{I}$ with the refinement is a lattice,
with minimal and maximal elements $\bot=\sset{\set{i}}{i\in L}$ and $\top=\set{L}$,
and the least upper and greatest lower bounds, $\vee$ and $\wedge$, can be constructed \cite{Roman-2008,Stanley-2012}.

With respect to the partitions $\xi\in P_\text{I}$,
we can define the partial correlation and entanglement properties,
as well as the measures quantifying them.

The \emph{$\xi$-uncorrelated states} are those which are products
with respect to the partition $\xi$,
\begin{subequations}
\label{eq:DI}
\begin{equation}
\label{eq:DuncI}
\begin{split}
&\mathcal{D}_{\xi\text{-unc}}:=\\
&\Bigsset{\varrho_L\in\mathcal{D}_L}{\forall X\in\xi, \exists \varrho_X\in\mathcal{D}_X:\varrho_L=\bigotimes_{X\in\xi}\varrho_X};
\end{split}
\end{equation}
the others are \emph{$\xi$-correlated states}.
With respect to the finest partition $\xi=\bot$, we call
$\bot$-uncorrelated states $\bigotimes_{i\in L}\varrho_i$ simply uncorrelated states.
The expectation values of all $\xi$-local observables factorize if and only if the state is $\xi$-uncorrelated.
The $\xi$-uncorrelated states are exactly those 
 which can be prepared from uncorrelated states by $\xi$-local operations.
(This is abbreviated as $\xi$-LO. With respect to the finest partition $\xi=\bot$, we write simply LO.)
The \emph{$\xi$-separable states} are 
convex combinations (or statistical mixtures) of $\xi$-uncorrelated states,
\begin{equation}
\label{eq:DsepI}
\mathcal{D}_{\xi\text{-sep}}:=\Conv\mathcal{D}_{\xi\text{-unc}}
\equiv\Bigsset{\sum_j p_j\varrho_j}{\varrho_j\in \mathcal{D}_{\xi\text{-unc}}};
\end{equation}
\end{subequations}
the others are \emph{$\xi$-entangled states}.
(Here, and in the entire text, 
$\sset{p_j}{j=1,2,\dots,m}$ is a probability distribution, 
$p_j\geq0$ and $\sum_j p_j=1$.)
The $\xi$-separable states are exactly those
 which can be prepared from uncorrelated states 
by $\xi$-local operations and classical communication among the parts $X\in\xi$.
(This is abbreviated as $\xi$-LOCC. With respect to the finest partition $\bot$, we write simply LOCC.)
Another point of view is that 
$\xi$-separable states are exactly those 
which can be prepared from uncorrelated states 
by mixtures of $\xi$-LOs.
We also have that
$\mathcal{D}_{\xi\text{-unc}}$ is closed under $\xi$-LO,
$\mathcal{D}_{\xi\text{-sep}}$ is closed under $\xi$-LOCC \cite{Szalay-2015b,Szalay-2018}.
Clearly, if a state is product with respect to a partition,
then it is also product with respect to any coarser partition,
it is always free to forget about some tensor product signs.
This means that these properties 
show the same lattice structure as the partitions \cite{Szalay-2015b,Szalay-2018}, $P_\text{I}$,
that is,
\begin{subequations}
\label{eq:oisomDI}
\begin{align}
\label{eq:oisomDuncI}
\upsilon\finereq\xi &\dspiff \mathcal{D}_{\upsilon\text{-unc}}\subseteq\mathcal{D}_{\xi\text{-unc}},\\
\label{eq:oisomDsepI}
\upsilon\finereq\xi &\dspiff \mathcal{D}_{\upsilon\text{-sep}}\subseteq\mathcal{D}_{\xi\text{-sep}}.
\end{align}
\end{subequations}

\begin{figure}\centering
\includegraphics{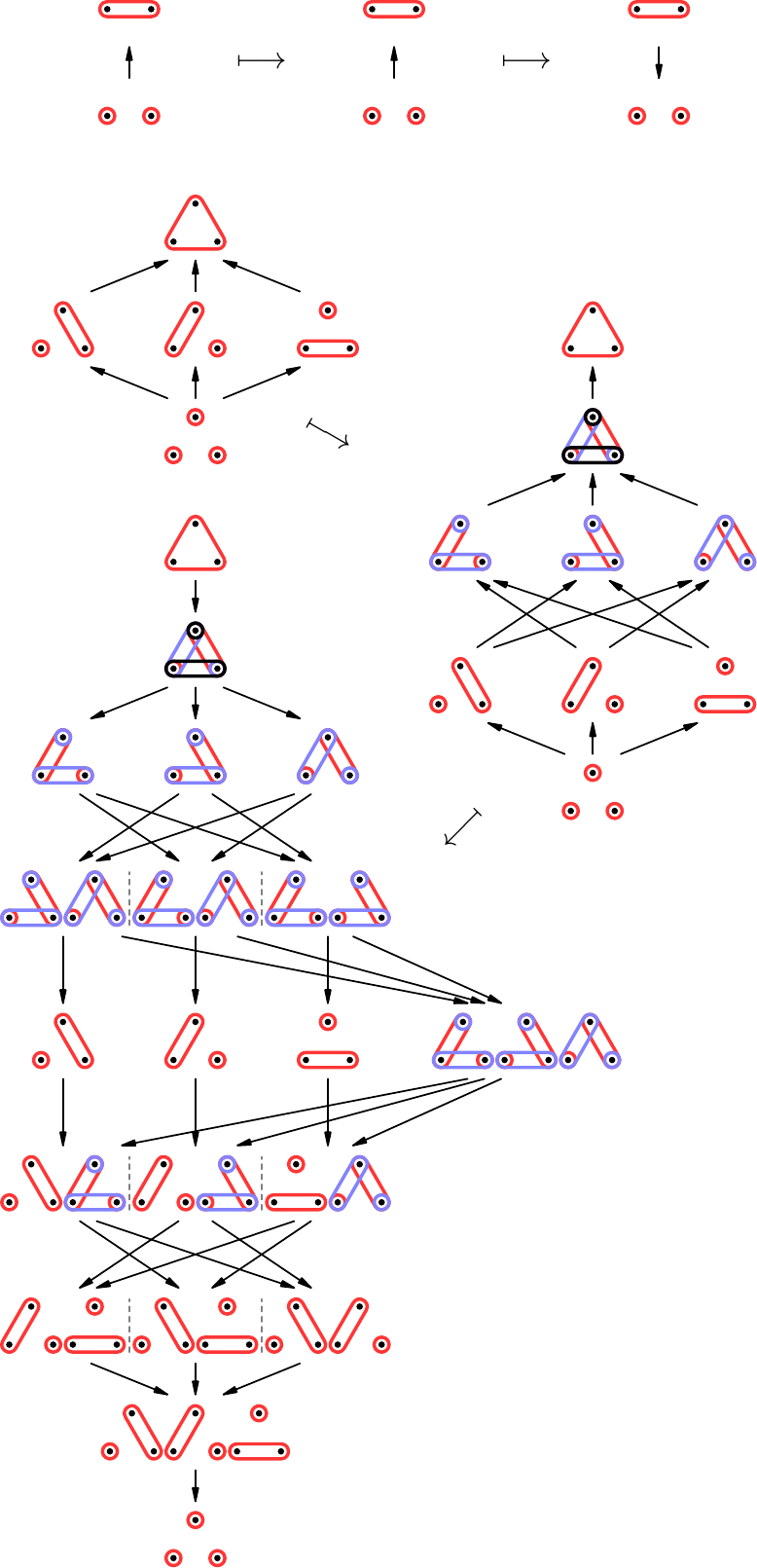}
\caption{The lattices of the three-level structure of multipartite correlation and entanglement for $n=2$ and $3$.
Only the maximal elements of the down-sets of $P_\text{I}$ are shown (with different colors) in $P_\text{II}$, while
only the minimal elements of the up-sets of $P_\text{II}$ are shown (side by side) in $P_\text{III}$.
The partial orders \eqref{eq:poI}, \eqref{eq:poII} and \eqref{eq:poIII} are represented by consecutive arrows.}\label{fig:Ps23}
\end{figure}

One can define the corresponding (information-geometry based) 
correlation and entanglement measures \cite{Szalay-2015b,Szalay-2017} 
for all $\xi$-correlation and $\xi$-entanglement.
These are the most natural generalizations of
the \emph{mutual information} \cite{Petz-2008,Wilde-2013}, 
the \emph{entanglement entropy} \cite{Bennett-1996a},
and the \emph{entanglement of formation} \cite{Bennett-1996b}
for Level~I of the multipartite case.

The \emph{$\xi$-correlation} of a state $\varrho$ is
its distinguishability by the relative entropy \eqref{eq:D} 
from the $\xi$-uncorrelated states 
\cite{Herbut-2004,Modi-2010,Szalay-2015b,Szalay-2017},
\begin{subequations}
\label{eq:measI}
\begin{equation}
\label{eq:CI}
C_\xi(\varrho) := 
\min_{\sigma\in \mathcal{D}_{\xi\text{-unc}}}D(\varrho||\sigma)
 = \sum_{X\in\xi}S(\varrho_X) - S(\varrho),
\end{equation}
given in terms of the von Neumann entropies \eqref{eq:S} of the
reduced, or marginal states $\varrho_X=\Tr_{L\setminus X}\varrho$ of subsystems $X\in\xi$.
Pure states of classical systems are always uncorrelated,
the correlation in pure states is of quantum origin, this is what we call entanglement \cite{Schrodinger-1935a,Schrodinger-1935b,Werner-1989,Horodecki-2009,Szalay-2015b,Szalay-2017,Szalay-2018}.
So, for pure states, the measure of entanglement should be that of correlation.
The \emph{$\xi$-entanglement} of a pure state $\pi=\cket{\psi}\bra{\psi}\in\mathcal{D}_L$ is
its $\xi$-correlation,
\begin{equation}
\label{eq:EpI}
E_\xi(\pi) := C_\xi \big\vert_\text{pure}(\pi) = {\sum}_{X\in\xi}S(\pi_X),
\end{equation}
with $\pi_X=\Tr_{L\setminus X}\pi$,
and for mixed states, one can use, for example, the convex roof extension \cite{Bennett-1996b,Uhlmann-1998,Uhlmann-2010}
to define the \emph{$\xi$-entanglement of formation}
\begin{equation}
\label{eq:EI}
\begin{split}
&E_\xi(\varrho) := \cnvroof{\bigl(C_\xi \big\vert_\text{pure}\bigr)}(\varrho)\\
&\equiv \min \Bigsset{{\sum}_j p_jC_\xi \big\vert_\text{pure}(\pi_j)}{{\sum}_j p_j\pi_j=\varrho},
\end{split}
\end{equation}
\end{subequations}
where ${\sum}_j p_j\pi_j$ is a pure decomposition of the state \cite{Szalay-2015b}.
We have that
$C_\xi$ is a correlation monotone (not increasing with respect to $\xi$-LO, for the proof see Appendix~\ref{app:misc.xiLOmon}),
$E_\xi$ is a strong entanglement monotone (convex and not increasing on average with respect to selective $\xi$-LOCC \cite{Vidal-2000,Horodecki-2001,Szalay-2015b}),
and both of these are faithful,
$C_\xi(\varrho)=0 \inliff \varrho\in\mathcal{D}_{\xi\text{-unc}}$,
$E_\xi(\varrho)=0 \inliff \varrho\in\mathcal{D}_{\xi\text{-sep}}$ \cite{Szalay-2015b},
moreover, they show the same lattice structure as the partitions, $P_\text{I}$,
that is,
\begin{subequations}
\label{eq:mmI}
\begin{align}
\label{eq:CmmI}
\upsilon\finereq\xi &\dspiff C_\upsilon\geq C_\xi,\\
\label{eq:EmmI}
\upsilon\finereq\xi &\dspiff E_\upsilon\geq E_\xi,
\end{align}
\end{subequations}
which is called \emph{multipartite monotonicity} \cite{Szalay-2015b,Szalay-2017}.

Note that, for two subsystems,
for the only nontrivial partition $1|2$, we have that $E_{1|2}(\pi) = 2S(\pi_1) = 2S(\pi_2)$ is just two times the usual \emph{entanglement entropy}.
However, we have a way of derivation completely different than the usual,
based on asymptotic LOCC convertibility from Bell-pairs \cite{Bennett-1996a}.
The usual way cannot be generalized to the multipartite scenario
(there is no ``reference state'', which was the Bell-pair in the bipartite scenario),
while our correlational, or statistical physical approach above could straightforwardly be generalized to the multipartite scenario
not only here, but also for the Level~II properties in the next subsystem.

\subsection{Level~II: multiple set partitions}
\label{sec:general.LII}

In multipartite entanglement theory,
it is necessary to handle mixtures of states uncorrelated with respect to different partitions~\cite{Acin-2001,Seevinck-2008,Szalay-2015b}.
For example, there are tripartite states which 
cannot be written as a mixture of a given kind of, e.g., $12|3$-uncorrelated states,
that is, not $12|3$-separable,
while can be written as a mixture of $12|3$-uncorrelated and $13|2$-uncorrelated states.
Such states should not be considered fully tripartite-entangled,
since there is no need for genuine tripartite entangled states in the mixture~\cite{Acin-2001,Seevinck-2008,Szalay-2012,Szalay-2015b}.
Also, if a state is $12|3$-separable and also $13|2$-separable,
it is not necessarily $1|2|3$-separable.
On the other hand, 
the order isomorphisms~\eqref{eq:oisomDuncI}-\eqref{eq:oisomDsepI} tell us that
if we consider states uncorrelated (or separable) with respect to a partition,
then we automatically consider states uncorrelated (or separable) with respect to all finer partitions.

To embed these requirements in the labeling of the multipartite correlation and entanglement properties,
we use the \emph{nonempty down-sets (nonempty ideals) of partitions} \cite{Szalay-2015b},
which are sets of partitions 
\begin{subequations}
\begin{equation}
\vs{\xi}=\set{\xi_1,\xi_2,\dots,\xi_{\abs{\vs{\xi}}}}\subseteq P_\text{I}
\end{equation}
which are closed downwards with respect to $\preceq$.
That is, if $\xi\in\vs{\xi}$, then for all partitions $\upsilon\finereq\xi$ we have $\upsilon\in\vs{\xi}$.
The set of the nonempty partition ideals 
is denoted with 
\begin{equation}
\label{eq:PII}
P_\text{II} :=\mathcal{O}_\downarrow(P_\text{I})\setminus\set{\emptyset}.
\end{equation}
This possesses a lattice structure with respect to the standard inclusion
as partial order, which we call \emph{refinement} again, and denote with $\preceq$ again,
\begin{equation}
\label{eq:poII}
\vs{\upsilon}\finereq\vs{\xi} \dspdef \vs{\upsilon}\subseteq\vs{\xi}.
\end{equation}
\end{subequations}
(For illustrations, see Figure~\ref{fig:Ps23}.)
The least upper and greatest lower bounds, $\vee$ and $\wedge$,
are then the union and the intersection, $\cup$ and $\cap$.

With respect to the partition ideals $\xi\in P_\text{II}$,
we can define the partial correlation and entanglement properties,
as well as the measures quantifying these.

The \emph{$\vs{\xi}$-uncorrelated states} are those which are $\xi$-uncorrelated \eqref{eq:DuncI} 
with respect to a $\xi\in\vs{\xi}$,
\begin{subequations}
\label{eq:DII}
\begin{equation}
\label{eq:DuncII}
\mathcal{D}_{\vs{\xi}\text{-unc}}
:= \bigcup_{\xi\in\vs{\xi}}\mathcal{D}_{\xi\text{-unc}};
\end{equation}
the others are \emph{$\vs{\xi}$-correlated states}.
The $\vs{\xi}$-uncorrelated states are exactly those
 which can be prepared from uncorrelated states by $\xi$-LO for a partition $\xi\in\vs{\xi}$.
The \emph{$\vs{\xi}$-separable states} are convex combinations of $\vs{\xi}$-uncorrelated states,
\begin{equation}
\label{eq:DsepII}
\mathcal{D}_{\vs{\xi}\text{-sep}}
:= \Conv\mathcal{D}_{\vs{\xi}\text{-unc}};
\end{equation}
\end{subequations}
the others are \emph{$\vs{\xi}$-entangled states}.
The $\vs{\xi}$-separable states are exactly those
which can be prepared from uncorrelated states 
by mixtures of $\xi$-LOs for different partitions $\xi\in\vs{\xi}$.
(Note that such transformations do not form a semigroup.) 
We also have that
$\mathcal{D}_{\vs{\xi}\text{-unc}}$ is closed under LO,
$\mathcal{D}_{\vs{\xi}\text{-sep}}$ is closed under LOCC \cite{Szalay-2015b,Szalay-2018}.
It follows from \eqref{eq:oisomDI} that
these properties show the same lattice structure as the partition ideals \cite{Szalay-2015b}, $P_\text{II}$,
that is,
\begin{subequations}
\label{eq:oisomDII}
\begin{align}
\label{eq:oisomDuncII}
\vs{\upsilon}\finereq\vs{\xi} &\dspiff \mathcal{D}_{\vs{\upsilon}\text{-unc}}\subseteq\mathcal{D}_{\vs{\xi}\text{-unc}},\\
\label{eq:oisomDsepII}
\vs{\upsilon}\finereq\vs{\xi} &\dspiff \mathcal{D}_{\vs{\upsilon}\text{-sep}}\subseteq\mathcal{D}_{\vs{\xi}\text{-sep}}.
\end{align}
\end{subequations}

One can define the corresponding (information-geometry based) 
correlation and entanglement measures \cite{Szalay-2015b,Szalay-2017} 
for all $\vs{\xi}$-correlation and $\vs{\xi}$-entanglement.
These are the most natural generalizations of
the \emph{mutual information} \cite{Petz-2008,Wilde-2013}, 
the \emph{entanglement entropy} \cite{Bennett-1996a},
and the \emph{entanglement of formation} \cite{Bennett-1996b}
for Level~II of the multipartite case.

The \emph{$\vs{\xi}$-correlation} of a state $\varrho$ is
its distinguishability by the relative entropy \eqref{eq:D}
from the $\vs{\xi}$-uncorrelated states \cite{Szalay-2015b,Szalay-2017},
\begin{subequations}
\label{eq:measII}
\begin{equation}
\label{eq:CII}
C_{\vs{\xi}}(\varrho) := 
\min_{\sigma\in \mathcal{D}_{\vs{\xi}\text{-unc}}}D(\varrho||\sigma) = \min_{\xi\in\vs{\xi}}C_{\xi}(\varrho).
\end{equation}
With the same reasoning as in Section~\ref{sec:general.LI},
the \emph{$\vs{\xi}$-entanglement} of a pure state is
\begin{equation}
\label{eq:EpII}
E_{\vs{\xi}}(\pi) := C_{\vs{\xi}} \big\vert_\text{pure}(\pi), 
\end{equation}
and for mixed states, one can use the convex roof extension
to define the \emph{$\vs{\xi}$-entanglement of formation}
\begin{equation}
\label{eq:EII}
E_{\vs{\xi}}(\varrho) :=
\cnvroof{\bigl(C_{\vs{\xi}} \big\vert_\text{pure}\bigr)}(\varrho).
\end{equation}
\end{subequations}
We have that
$C_{\vs{\xi}}$ is a correlation monotone (for the proof, see Appendix~\ref{app:misc.xiLOmon}),
$E_{\vs{\xi}}$ is a strong entanglement monotone \cite{Szalay-2015b},
and both of these are faithful,
$C_{\vs{\xi}}(\varrho)=0 \inliff \varrho\in\mathcal{D}_{\vs{\xi}\text{-unc}}$,
$E_{\vs{\xi}}(\varrho)=0 \inliff \varrho\in\mathcal{D}_{\vs{\xi}\text{-sep}}$ \cite{Szalay-2015b},
moreover, they show the same lattice structure as the partition ideals, $P_\text{II}$,
that is,
\begin{subequations}
\label{eq:mmII}
\begin{align}
\label{eq:CmmII}
\vs{\upsilon}\finereq\vs{\xi} &\dspiff C_{\vs{\upsilon}}\geq C_{\vs{\xi}},\\
\label{eq:EmmII}
\vs{\upsilon}\finereq\vs{\xi} &\dspiff E_{\vs{\upsilon}}\geq E_{\vs{\xi}},
\end{align}
\end{subequations}
which is called \emph{multipartite monotonicity} for Level~II \cite{Szalay-2015b,Szalay-2017}.

\subsection{Level~III: classes}
\label{sec:general.LIII}

The partial correlation and entanglement properties form an inclusion hierarchy \eqref{eq:oisomDuncII}-\eqref{eq:oisomDsepII},
for example, if a state is $1|2|3$-separable, then it is also $12|3$-separable.
We are interested in the labeling of the strict, or exclusive properties,
for example, those states which are $12|3$-separable and not $1|2|3$-separable.
In general, we would like to determine all the possible nonempty intersections of the state sets
$\mathcal{D}_{\vs{\xi}\text{-unc}}$ and $\mathcal{D}_{\vs{\xi}\text{-sep}}$.
We call these intersections \emph{partial correlation} and \emph{partial entanglement classes,}
containing states of well-defined partial correlation and partial entanglement properties.

To embed these requirements in the labeling of the strict properties of multipartite correlation and entanglement,
we use the \emph{nonempty up-sets (nonempty filters) of nonempty down-sets (nonempty ideals) of partitions} \cite{Szalay-2015b}, which are
sets of down-sets of partitions
\begin{subequations}
\begin{equation}
\vvs{\xi}=\set{\vs{\xi}_1,\vs{\xi}_2,\dots,\vs{\xi}_{\abs{\vvs{\xi}}}}\subseteq P_\text{II} 
\end{equation}
which are closed upwards with respect to $\preceq$.
That is, if $\vs{\xi}\in\vvs{\xi}$, then for all partition ideals $\vs{\upsilon}\coarsereq\vs{\xi}$ we have $\vs{\upsilon}\in\vvs{\xi}$.
The set of the nonempty up-sets of nonempty down-sets of partitions of $L$ is denoted with 
\begin{equation}
\label{eq:PIII}
P_\text{III} :=\mathcal{O}_\uparrow(P_\text{II})\setminus\set{\emptyset}.
\end{equation}
This possesses a lattice structure with respect to the standard inclusion
as partial order, which we call \emph{refinement} again, and denote with $\preceq$ again,
\begin{equation}
\label{eq:poIII}
\vvs{\upsilon}\finereq\vvs{\xi} \dspdef \vvs{\upsilon}\subseteq\vvs{\xi}.
\end{equation}
\end{subequations}
(For illustrations, see Figure~\ref{fig:Ps23}.)
Note that here, contrary to Level~I and II, we call $\vvs{\xi}$ finer and $\vvs{\upsilon}$ coarser.
In the generic case,
if the inclusion of sets is described by a poset $P$,
then the possible intersections can be described by $\mathcal{O}_\uparrow(P)$
(for the proof, see appendix A.2 in \cite{Szalay-2018}).
One may make the classification coarser
by selecting a subposet $P_\text{II*}\subseteq P_\text{II}$,
with respect to which the classification is done,
$P_\text{III*} :=\mathcal{O}_\uparrow(P_\text{II*})\setminus\set{\emptyset}$ \cite{Szalay-2015b,Szalay-2018}.

With respect to the partition ideal filters $\vvs{\xi}\in P_\text{III*}$,
we can define the strict partial correlation and entanglement properties.

The \emph{strictly $\vvs{\xi}$-uncorrelated states} are those which are
uncorrelated with respect to all $\vs{\xi} \in\vvs{\xi}$,
and correlated with respect to all $\vs{\xi}\in\cmpl{\vvs{\xi}}=P_\text{II*}\setminus\vvs{\xi}$,
so the \emph{class} of these 
(\emph{partial correlation class}) is
\begin{subequations}
\label{eq:CIII}
\begin{equation}
\label{eq:CuncIII}
\mathcal{C}_{\vvs{\xi}\text{-unc}} :=
 \bigcap_{\vs{\xi}\in\cmpl{\vvs{\xi}}} \cmpl{\mathcal{D}_{\vs{\xi}\text{-unc}}} \cap 
 \bigcap_{\vs{\xi}\in      \vvs{\xi} }       \mathcal{D}_{\vs{\xi}\text{-unc}}.
\end{equation}
The \emph{strictly $\vvs{\xi}$-separable states} are those which are
separable with respect to all $\vs{\xi} \in\vvs{\xi}$,
and entangled with respect to all $\vs{\xi}\in\cmpl{\vvs{\xi}}$,
so the \emph{class} of these 
(\emph{partial separability class}, or \emph{partial entanglement class}) is
\begin{equation}
\label{eq:CsepIII}
\mathcal{C}_{\vvs{\xi}\text{-sep}} :=
 \bigcap_{\vs{\xi}\in\cmpl{\vvs{\xi}}} \cmpl{\mathcal{D}_{\vs{\xi}\text{-sep}}} \cap 
 \bigcap_{\vs{\xi}\in      \vvs{\xi} }       \mathcal{D}_{\vs{\xi}\text{-sep}}.
\end{equation}
\end{subequations}
The meaning of the Level~III hierarchy could also be clarified.
If there exists a $\varrho \in \mathcal{C}_{\vvs{\upsilon}\text{-unc}}$ and an LO
mapping it into $\mathcal{C}_{\vvs{\xi}\text{-unc}}$, then $\vvs{\upsilon}\finereq\vvs{\xi}$ 
\cite{Szalay-2018,Szalay-2015b};
and
if there exists a $\varrho \in \mathcal{C}_{\vvs{\upsilon}\text{-sep}}$ and an LOCC
mapping it into $\mathcal{C}_{\vvs{\xi}\text{-sep}}$, then $\vvs{\upsilon}\finereq\vvs{\xi}$ 
\cite{Szalay-2015b}.
In this sense, the Level~III hierarchy compares the strength of
correlation and entanglement among the classes.

The filters $\vvs{\xi}$ are \emph{sufficient} for the description of the 
nontrivial classes (nonempty intersections)
of the correlation and entanglement properties,
but they are \emph{not necessary} in general.
For the strictly $\vvs{\xi}$-uncorrelated states,
the structure $P_\text{III*}$ simplifies significantly \cite{Szalay-2018}.
(For example, for the finest classification $P_\text{II*}= P_\text{II}$,
we have that the structure of the partial correlation classes is
$P_\text{III*}\cong P_\text{I}^\partial$, 
as the unique labeling of the nontrivial partial correlation classes is given as
$\vvs{\xi}=\upset\set{\downset\set{\xi}}$ for the partitions $\xi\in P_\text{I}$ \cite{Szalay-2018}.)
On the other hand, it is still a conjecture that
$\vvs{\xi}$-separability is nontrivial for all $\vvs{\xi}\in P_\text{III}$ \cite{Szalay-2015b}.
The conjecture holds for three subsystems, for which explicit examples were constructed \cite{Szalay-2012,Han-2018}.

\section{Multipartite correlation and entanglement: permutation invariant properties}
\label{sec:perminv}

Here we build up the structure of the classification and quantification 
of the permutation invariant multipartite correlation and entanglement properties.
We do this from the point of view of the general properties in the previous section.  
To give a quick guidance on the construction, 
we show the diagram
\begin{equation}
\label{eq:cd2}
\xymatrix@M+=8bp{
(P_\text{III},\preceq) \ar@{|->}[r]^s                                                               &
  (\pinv{P}_\text{III},\preceq)              \\
(P_\text{II}, \preceq) \ar@{|->}[r]^s \ar@{|->}[u]_{\mathcal{O}_\uparrow  \setminus\set{\emptyset}} &
  (\pinv{P}_\text{II}, \preceq) \ar@{|->}[u]_{\mathcal{O}_\uparrow  \setminus\set{\emptyset}} \\ 
(P_\text{I},  \preceq) \ar@{|->}[r]^s \ar@{|->}[u]_{\mathcal{O}_\downarrow\setminus\set{\emptyset}} &
  (\pinv{P}_\text{I},  \preceq) \ar@{|->}[u]_{\mathcal{O}_\downarrow\setminus\set{\emptyset}} 
}
\end{equation}
in advance.
This might not seem to be obvious, for the proof, see Corollary \ref{cor:commut} in Appendix~\ref{app:structPI.part}.

In Section~\ref{sec:general}, for the description of the structure of multipartite correlations \cite{Szalay-2015b,Szalay-2017,Szalay-2018},
the natural language was that of \emph{set partitions} \cite{Roman-2008}.
For the permutation-invariant case,
particularly $k$-partitionability and $k$-producibility,
it is natural to use the simpler language of \emph{integer partitions} \cite{Andrews-1984,Stanley-2012}.

\subsection{Level 0: subsystem sizes}
\label{sec:perminv.L0p}

Let us have the map, simply measuring the sizes of the subsystems $X\in P_0$,
\begin{equation}
\label{eq:s0}
s(X) := \abs{X},
\end{equation}
mapping $P_\text{0}\to \pinv{P}_\text{0}$,
which is simply
\begin{equation}
\label{eq:P0p}
\pinv{P}_\text{0} := s(P_\text{0}) = \set{0,1,2,\dots,n}.
\end{equation}
This possesses the natural order $\leq$, which is now denoted with $\finereq$,
with respect to which $s$ is monotone,
$Y\preceq X \inlthen s(Y)\preceq s(X)$.  
We note for the sake of the analogy with the construction in the subsequent levels
that one could define $\preceq$ on $\pinv{P}_\text{0}$ as
$y\preceq x$ if there exist $Y\in s^{-1}(y)$ and $X\in s^{-1}(x)$
such that $Y\preceq X$, from which the monotonicity of $s$ follows automatically
(see \eqref{eq:mon} in Appendix~\ref{app:structPI.coars}).

\subsection{Level~I: integer partitions}
\label{sec:perminv.LIp}

Let the map $s$ in \eqref{eq:s0} act elementwisely on the partitions $\xi\in P_\text{I}$,
\begin{equation}
\label{eq:sI}
s(\xi) := \bigsmset{s(X)}{\forall X\in \xi},
\end{equation}
where on the right-hand side a \emph{multiset} stays,
that is,
a set, which allows multiple instances of its elements.
(For example, for $\xi=\mset{\set{1},\set{2,4},\set{3}}$,
$s(\xi)=\mset{1,2,1}\equiv\mset{2,1,1}$.)
We call the integer partition $s(\xi)$ the \emph{type} of the set partition $\xi$.
This is then an \emph{integer partition} of $n$ \cite{Andrews-1984,Stanley-2012}, which is denoted as
\begin{subequations}
\begin{equation}
\pinv{\xi}=\mset{x_1,x_2,\dots,x_{\abs{\pinv{\xi}}}},
\end{equation}
where the parts $x\in\pinv{\xi}$ are nonzero subsystem sizes,
for which $\sum_{x\in\pinv{\xi}} x = n$.
The set of the integer partitions of $n$ is denoted with
\begin{equation}
\label{eq:PIp}
\pinv{P}_\text{I} := s(P_\text{I}).
\end{equation}
Its size is given by the \emph{partition numbers} \cite{oeisA000041}.
Based on the refinement of the set partitions \eqref{eq:poI}, 
we define a partial order over the integer partitions,
which we call \emph{refinement} again, and denote with $\preceq$ agai again,
given as
\begin{equation}
\label{eq:poIp}
\pinv{\upsilon}\finereq\pinv{\xi} \dspdef \exists \upsilon\in s^{-1}(\pinv{\upsilon}),\xi\in s^{-1}(\pinv{\xi}) \dispt{s.t.} \upsilon\finereq\xi.
\end{equation}
\end{subequations}
(For illustrations, see Figures~\ref{fig:PPpI234enc} and \ref{fig:Pp56}.
The integer partitions are represented by their Young diagrams, where the rows of boxes are ordered decreasingly.)
This is a proper partial order (for the proof, see Appendix~\ref{app:structPI.part}).
Note that the refinement is a partial order only,
there are pairs of partitions which cannot be ordered,
for example,
$\mset{2,2} \npreceq \mset{3,1}$ and
$\mset{2,2} \nsucceq \mset{3,1}$, see Figure~\ref{fig:PPpI234enc}.
Note that $\pinv{P}_\text{I}$ with this partial order is \emph{not} a lattice,
since there are no unique least upper and greatest lower bounds with respect to that.
(The first example is for $n=5$, where
$\mset{2,2,1},\mset{3,1,1}\preceq\mset{3,2},\mset{4,1}$,
and there is no integer partition between them, see Figure~\ref{fig:Pp56}.)
The refinement also admits a minimal and a maximal element $\bot=\mset{1,1,\dots,1}$ and $\top=\mset{n}$.
By construction,
$s$ is monotone with respect to these partial orders,
$\upsilon\preceq\xi \inlthen s(\upsilon)\preceq s(\xi)$
(see \eqref{eq:mon} in Appendix~\ref{app:structPI.coars}).

\begin{figure*}\centering
\includegraphics{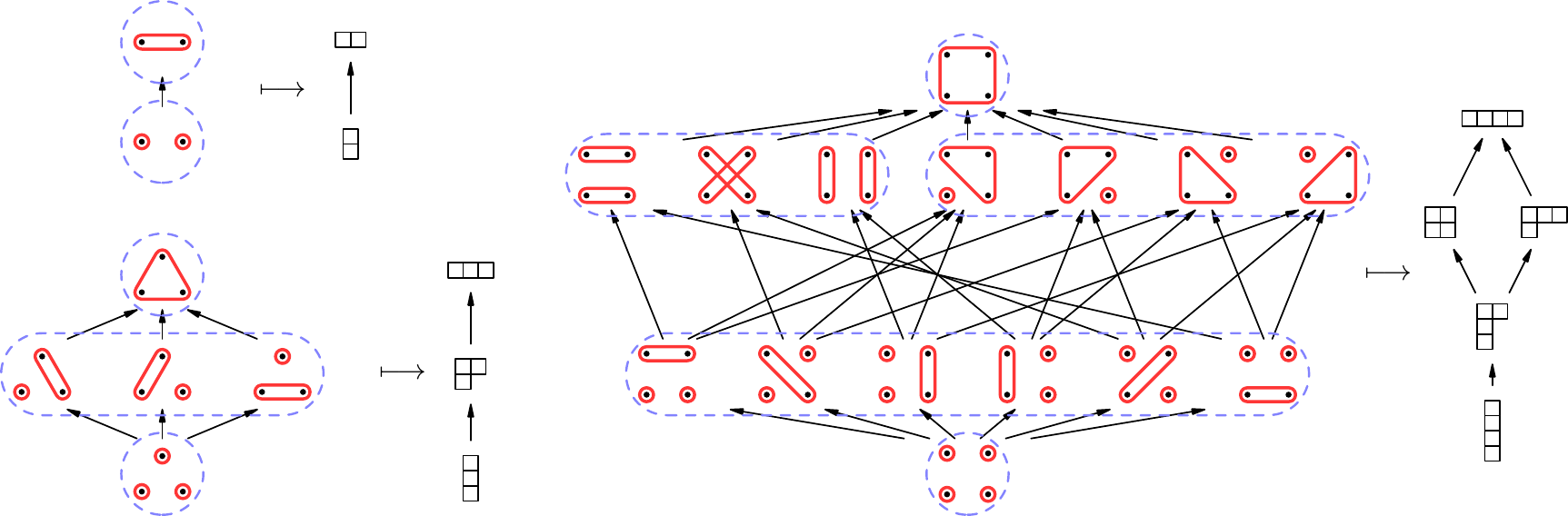}
\caption{
The construction of the posets $\pinv{P}_\text{I}$ of permutation invariant multipartite correlation and entanglement properties
from the lattices $P_\text{I}$ for $n=2,3$ and $4$.
The integer partitions \eqref{eq:sI} are represented by their Young diagrams,
the partial order \eqref{eq:poIp} is represented by consecutive arrows.}\label{fig:PPpI234enc}
\end{figure*}

\begin{figure}\centering
\includegraphics{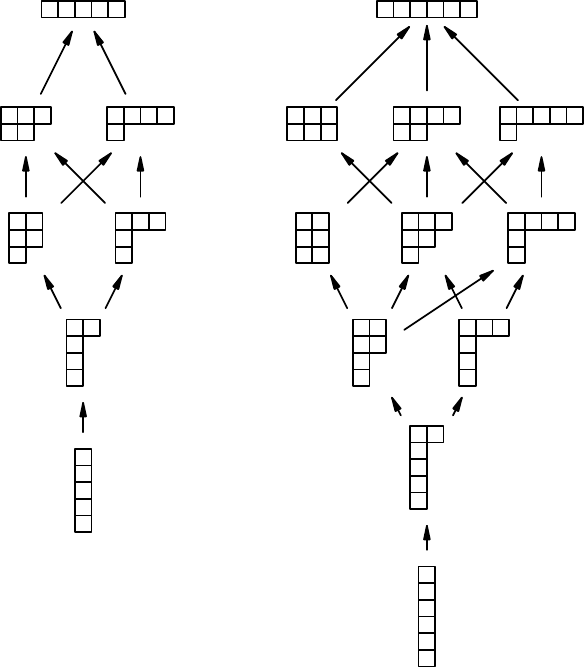}
\caption{
The posets $\pinv{P}_\text{I}$ of permutation invariant multipartite correlation and entanglement properties for $n=5$ and $6$.
The notation is the same as in Figure~\ref{fig:PPpI234enc}.
}\label{fig:Pp56}
\end{figure}

With respect to the integer partitions $\pinv{\xi}\in \pinv{P}_\text{I}$,
we can define the partial correlation and entanglement properties,
as well as the measures quantifying them.

The \emph{$\pinv{\xi}$-uncorrelated states} are those which are $\xi$-uncorrelated \eqref{eq:DuncI}
with respect to a partition $\xi$ of type $\pinv{\xi}$,
which is actually a Level~II state set \eqref{eq:DuncII} labeled by the ideal $\vs{\xi}=\downset s^{-1}(\pinv{\xi})$,
\begin{subequations}
\label{eq:DIp}
\begin{equation}
\label{eq:DuncIp}
\mathcal{D}_{\pinv{\xi}\text{-unc}}
:= \bigcup_{\xi: s(\xi) = \pinv{\xi}}\mathcal{D}_{\xi\text{-unc}}
 = \mathcal{D}_{\downset s^{-1}(\pinv{\xi})\text{-unc}};
\end{equation}
the others are \emph{$\pinv{\xi}$-correlated states}.
The $\pinv{\xi}$-uncorrelated states are exactly those
 which can be prepared from uncorrelated states by $\xi$-LO for a partition $\xi$ of type $\pinv{\xi}$.
(Note that we consider arbitrary states here, not only permutation invariant ones.
The correlation property defined in this way is what permutation invariant is.)
The \emph{$\pinv{\xi}$-separable states} are 
convex combinations of $\pinv{\xi}$-uncorrelated states,
which is actually a Level~II state set \eqref{eq:DsepII} labeled by the ideal $\vs{\xi}=\downset s^{-1}(\pinv{\xi})$,
\begin{equation}
\label{eq:DsepIp}
\mathcal{D}_{\pinv{\xi}\text{-sep}}
:= \Conv\mathcal{D}_{\pinv{\xi}\text{-unc}}
 = \mathcal{D}_{\downset s^{-1}(\pinv{\xi})\text{-sep}};
\end{equation}
\end{subequations}
the others are \emph{$\pinv{\xi}$-entangled states}.
The $\pinv{\xi}$-separable states are exactly those
which can be prepared from uncorrelated states 
by mixtures of $\xi$-LOs for different partitions $\xi$ of type $\pinv{\xi}$. 
We also have that
$\mathcal{D}_{\pinv{\xi}\text{-unc}}$ is closed under LO,
$\mathcal{D}_{\pinv{\xi}\text{-sep}}$ is closed under LOCC,
because these hold for the general Level~II state sets 
$\mathcal{D}_{\downset s^{-1}(\pinv{\xi})\text{-unc}}$ and 
$\mathcal{D}_{\downset s^{-1}(\pinv{\xi})\text{-sep}}$, see in Section~\ref{sec:general.LII}.
It also follows that
these properties show the same partially ordered structure as the integer partitions, $\pinv{P}_\text{I}$,
that is,
\begin{subequations}
\label{eq:oisomDIp}
\begin{align}
\label{eq:oisomDuncIp}
\pinv{\upsilon}\finereq\pinv{\xi} &\dspiff \mathcal{D}_{\pinv{\upsilon}\text{-unc}}\subseteq\mathcal{D}_{\pinv{\xi}\text{-unc}},\\
\label{eq:oisomDsepIp}
\pinv{\upsilon}\finereq\pinv{\xi} &\dspiff \mathcal{D}_{\pinv{\upsilon}\text{-sep}}\subseteq\mathcal{D}_{\pinv{\xi}\text{-sep}}.
\end{align}
\end{subequations}
(These come by using \eqref{eq:embedPpIPII} and \eqref{eq:labelsPpIpr} in Appendix~\ref{app:structPI.part}, 
then \eqref{eq:oisomDII} and \eqref{eq:DIp}.)

In accordance with the general case in Section~\ref{sec:general.LI},
one can define the corresponding (information-geometry based)
correlation and entanglement measures
for all $\pinv{\xi}$-correlation and $\pinv{\xi}$-entanglement.
These are the most natural generalizations of
the \emph{mutual information} \cite{Petz-2008,Wilde-2013}, 
the \emph{entanglement entropy} \cite{Bennett-1996a},
and the \emph{entanglement of formation} \cite{Bennett-1996b}
for Level~I of the permutation invariant multipartite case.

The \emph{$\pinv{\xi}$-correlation} of a state $\varrho$ is
its distinguishability by the relative entropy \eqref{eq:D} 
from the $\pinv{\xi}$-uncorrelated states,
which is actually a Level~II measure \eqref{eq:CII} labeled by the ideal $\vs{\xi}=\downset s^{-1}(\pinv{\xi})$,
\begin{subequations}
\label{eq:measIp}
\begin{equation}
\label{eq:CIp}
C_{\pinv{\xi}}(\varrho) 
:= \min_{\sigma\in \mathcal{D}_{\pinv{\xi}\text{-unc}}}D(\varrho||\sigma) 
= C_{\downset s^{-1}(\pinv{\xi})}(\varrho).
\end{equation}
With the same reasoning as in Section~\ref{sec:general.LI},
the \emph{$\pinv{\xi}$-entanglement} of a pure state is
\begin{equation}
\label{eq:EpIp}
E_{\pinv{\xi}}(\pi) := C_{\pinv{\xi}} \big\vert_\text{pure}(\pi)
=E_{\downset s^{-1}(\pinv{\xi})}(\pi),
\end{equation}
and for mixed states, one can use the convex roof extension,
which is actually a Level~II measure \eqref{eq:EII} labeled by the ideal $\vs{\xi}=\downset s^{-1}(\pinv{\xi})$,
\begin{equation}
\label{eq:EIp}
E_{\pinv{\xi}}(\varrho) 
:= \cnvroof{\bigl(C_{\pinv{\xi}} \big\vert_\text{pure}\bigr)}(\varrho)
= E_{\downset s^{-1}(\pinv{\xi})}(\varrho).
\end{equation}
\end{subequations}
These inherit the properties of the general case in Section~\ref{sec:general.LII}.
So we have that
$C_{\pinv{\xi}}$ is a correlation monotone, 
$E_{\pinv{\xi}}$ is a strong entanglement monotone,
and both of these are faithful,
$C_{\pinv{\xi}}(\varrho)=0 \inliff \varrho\in\mathcal{D}_{\pinv{\xi}\text{-unc}}$,
$E_{\pinv{\xi}}(\varrho)=0 \inliff \varrho\in\mathcal{D}_{\pinv{\xi}\text{-sep}}$,
moreover, they show the same partially ordered structure as the integer partitions, $\pinv{P}_\text{I}$,
that is,
\begin{subequations}
\begin{align}
\label{eq:CmmIp}
\pinv{\upsilon}\finereq\pinv{\xi} &\dspiff C_{\pinv{\upsilon}}\geq C_{\pinv{\xi}},\\
\label{eq:EmmIp}
\pinv{\upsilon}\finereq\pinv{\xi} &\dspiff E_{\pinv{\upsilon}}\geq E_{\pinv{\xi}},
\end{align}
\end{subequations}
which we call \emph{multipartite monotonicity} for Level~I of the permutation invariant case.
(These come by using \eqref{eq:embedPpIPII} and \eqref{eq:labelsPpIpr} in Appendix~\ref{app:structPI.part}, 
then \eqref{eq:mmII} and \eqref{eq:measIp}.)

\subsection{Level~II: multiple integer partitions}
\label{sec:perminv.LIIp}

Let the map $s$ in \eqref{eq:sI} act elementwisely on the partition ideals $\vs{\xi}\in P_\text{II}$, 
\begin{equation}
\label{eq:sII}
s(\vs{\xi}) := \bigsset{s(\xi)}{\forall \xi\in \vs{\xi}},
\end{equation}
where on the right-hand side a \emph{set} stays.
We call $s(\vs{\xi})$ the \emph{type} of the set partition ideal $\vs{\xi}$.
This is a set of integer partitions, which is denoted as
\begin{subequations}
\begin{equation}
\pinv{\vs{\xi}} = \bigset{\pinv{\xi}_1,\pinv{\xi}_2,\dots,\pinv{\xi}_{\abs{\pinv{\vs{\xi}}}}} \subseteq \pinv{P}_\text{I},
\end{equation}
which is a nonempty ideal of integer partitions, the set of which is denoted as
\begin{equation}
\label{eq:PIIp}
\pinv{P}_\text{II} := s(P_\text{II})
=\mathcal{O}_\downarrow(\pinv{P}_\text{I})\setminus\set{\emptyset}.
\end{equation}
This might not seem to be obvious, this is given in \eqref{eq:cd2}, proven in Corollary \ref{cor:commut} in Appendix~\ref{app:structPI.part}.
In particular, all $\xi\in\vs{\xi}$ is of type $\pinv{\xi}$ for a $\pinv{\xi}\in s(\vs{\xi})$,
on the other hand, all $\pinv{\xi}\in s(\vs{\xi})$ has representative $\xi\in\vs{\xi}$ of type $\pinv{\xi}$.
Based on the refinement of the set partition ideals \eqref{eq:poII}, 
we define a partial order over the integer partition ideals,
which we call \emph{refinement} again, and denote with $\preceq$ again, 
given as
\begin{equation}
\label{eq:poIIp}
\begin{split}
\pinv{\vs{\upsilon}}\finereq\pinv{\vs{\xi}} &\dspdef
\exists \vs{\upsilon}\in s^{-1}(\pinv{\vs{\upsilon}}),\vs{\xi}\in s^{-1}(\pinv{\vs{\xi}}) \;\text{s.t.}\; \vs{\upsilon}\finereq\vs{\xi}\\
&\dspiff \pinv{\vs{\upsilon}}\subseteq \pinv{\vs{\xi}},
\end{split}
\end{equation}
\end{subequations}
which turns out to be the inclusion, being the natural partial order for ideals.
This might not seem to be obvious, this is given in \eqref{eq:cd2}, proven in Corollary \ref{cor:commut} in Appendix~\ref{app:structPI.part}.
(For illustrations, see Figure~\ref{fig:Pps2345}.)
Because of these, $\pinv{P}_\text{II}$ with this partial order \emph{is} a lattice.
By construction,
$s$ is monotone with respect to these partial orders,
$\vs{\upsilon}\preceq\vs{\xi} \inlthen s(\vs{\upsilon})\preceq s(\vs{\xi})$
(see \eqref{eq:mon} in Appendix~\ref{app:structPI.coars}).

With respect to the integer partition ideals $\pinv{\vs{\xi}}\in \pinv{P}_\text{II}$,
we can define the partial correlation and entanglement properties,
as well as the measures quantifying these.

The \emph{$\pinv{\vs{\xi}}$-uncorrelated states} are those 
which are $\xi$-uncorrelated \eqref{eq:DuncI}
with respect to a partition $\xi$ of type $\pinv{\xi}$ for a $\pinv{\xi}\in\pinv{\vs{\xi}}$,
which is actually a Level~II state set \eqref{eq:DuncII} labeled by the ideal 
$\vs{\xi}= \vee s^{-1}(\pinv{\vs{\xi}})= \bigvee_{\pinv{\xi}\in\pinv{\vs{\xi}}}\downset s^{-1}(\pinv{\xi})$, 
\begin{subequations}
\label{eq:DIIp}
\begin{equation}
\label{eq:DuncIIp}
\begin{split}
\mathcal{D}_{\pinv{\vs{\xi}}\text{-unc}}
:&= \bigcup_{\vs{\xi}: s(\vs{\xi}) = \pinv{\vs{\xi}}}\mathcal{D}_{\vs{\xi}\text{-unc}}
 = \mathcal{D}_{\vee s^{-1}(\pinv{\vs{\xi}})\text{-unc}}\\
 &= \bigcup_{\xi: s(\xi) \in \pinv{\vs{\xi}}}\mathcal{D}_{\xi\text{-unc}}
 = \mathcal{D}_{\bigvee_{\pinv{\xi}\in\pinv{\vs{\xi}}} \downset s^{-1}(\pinv{\xi})\text{-unc}},
\end{split}
\end{equation}
(for the proof, see \eqref{eq:labelsPpII} in Appendix~\ref{app:structPI.part});
the others are \emph{$\pinv{\vs{\xi}}$-correlated states}.
(We use the notation $\vee A := \bigvee_{a\in A} a$ for any subset $A$ of a lattice.)
The $\pinv{\vs{\xi}}$-uncorrelated states are exactly those
 which can be prepared from uncorrelated states by $\xi$-LO for a partition $\xi$ of a type $\pinv{\xi}\in\pinv{\vs{\xi}}$. 
The \emph{$\pinv{\vs{\xi}}$-separable states} are 
convex combinations of $\pinv{\vs{\xi}}$-uncorrelated states,
which is actually a Level~II state set \eqref{eq:DsepII} labeled by the ideal
$\vs{\xi}= \vee s^{-1}(\pinv{\vs{\xi}})= \bigvee_{\pinv{\xi}\in\pinv{\vs{\xi}}}\downset s^{-1}(\pinv{\xi})$,
\begin{equation}
\label{eq:DsepIIp}
\mathcal{D}_{\pinv{\vs{\xi}}\text{-sep}}
:= \Conv\mathcal{D}_{\pinv{\vs{\xi}}\text{-unc}}
 = \mathcal{D}_{\vee s^{-1}(\pinv{\vs{\xi}})\text{-sep}},
\end{equation}
\end{subequations}
the others are \emph{$\pinv{\vs{\xi}}$-entangled states}.
The $\pinv{\vs{\xi}}$-separable states are exactly those
which can be prepared from uncorrelated states 
by mixtures of $\xi$-LOs for different partitions $\xi$ of different types $\pinv{\xi}\in\pinv{\vs{\xi}}$.
We also have that
$\mathcal{D}_{\pinv{\vs{\xi}}\text{-unc}}$ is closed under LO,
$\mathcal{D}_{\pinv{\vs{\xi}}\text{-sep}}$ is closed under LOCC,
because these hold for the general Level~II state sets 
$\mathcal{D}_{\vee s^{-1}(\pinv{\vs{\xi}})}$ and
$\mathcal{D}_{\vee s^{-1}(\pinv{\vs{\xi}})}$, see in Section~\ref{sec:general.LII}.
It also follows that
these properties show the same lattice structure as the integer partition ideals, $\pinv{P}_\text{II}$,
that is,
\begin{subequations}
\label{eq:oisomDIIp}
\begin{align}
\label{eq:oisomDuncIIp}
\pinv{\vs{\upsilon}}\finereq\pinv{\vs{\xi}} 
&\dspiff \mathcal{D}_{\pinv{\vs{\upsilon}}\text{-unc}}\subseteq\mathcal{D}_{\pinv{\vs{\xi}}\text{-unc}},\\
\label{eq:oisomDsepIIp}
\pinv{\vs{\upsilon}}\finereq\pinv{\vs{\xi}}
&\dspiff \mathcal{D}_{\pinv{\vs{\upsilon}}\text{-sep}}\subseteq\mathcal{D}_{\pinv{\vs{\xi}}\text{-sep}}.
\end{align}
\end{subequations}
(These come by using \eqref{eq:embedPpIIPII} and \eqref{eq:labelsPpII} in Appendix~\ref{app:structPI.part}, 
then \eqref{eq:oisomDII} and \eqref{eq:DIIp}.)

In accordance with the general case in Section~\ref{sec:general.LII},
one can define the corresponding (information-geometry based)
correlation and entanglement measures
for all $\pinv{\vs{\xi}}$-correlation and $\pinv{\vs{\xi}}$-entanglement.
These are the most natural generalizations of
the \emph{mutual information} \cite{Petz-2008,Wilde-2013}, 
the \emph{entanglement entropy} \cite{Bennett-1996a},
and the \emph{entanglement of formation} \cite{Bennett-1996b}
for Level~II of the permutation invariant multipartite case.

The \emph{$\pinv{\vs{\xi}}$-correlation} of a state $\varrho$ is
its distinguishability by the relative entropy \eqref{eq:D}
from the $\pinv{\vs{\xi}}$-uncorrelated states, 
which is actually a Level~II measure \eqref{eq:CII} labeled by the ideal 
$\vs{\xi}=\vee s^{-1}(\pinv{\vs{\xi}})$, 
\begin{subequations}
\label{eq:measIIp}
\begin{equation}
\label{eq:CIIp}
C_{\pinv{\vs{\xi}}}(\varrho) 
:= \min_{\sigma\in \mathcal{D}_{\pinv{\vs{\xi}}\text{-unc}}}D(\varrho||\sigma) 
= C_{\vee s^{-1}(\pinv{\vs{\xi}})}(\varrho).
\end{equation}
With the same reasoning as in Section~\ref{sec:general.LII},
the \emph{$\pinv{\vs{\xi}}$-entanglement} of a pure state is
\begin{equation}
\label{eq:EpIIp}
E_{\pinv{\vs{\xi}}}(\pi) := C_{\pinv{\vs{\xi}}} \big\vert_\text{pure}(\pi)
=E_{\vee s^{-1}(\pinv{\vs{\xi}})}(\pi),
\end{equation}
and for mixed states, one can use the convex roof extension,
which is actually a Level~II measure \eqref{eq:EII} labeled by the ideal 
$\vs{\xi}=\vee s^{-1}(\pinv{\vs{\xi}})$,
\begin{equation}
\label{eq:EIIp}
E_{\pinv{\vs{\xi}}}(\varrho) 
:= \cnvroof{\bigl(C_{\pinv{\vs{\xi}}} \big\vert_\text{pure}\bigr)}(\varrho)
= E_{\vee s^{-1}(\pinv{\vs{\xi}})}(\varrho).
\end{equation}
\end{subequations}
These inherit the properties of of the general case in Section~\ref{sec:general.LII}.
So we have that
$C_{\pinv{\vs{\xi}}}$ is a correlation monotone, 
$E_{\pinv{\vs{\xi}}}$ is a strong entanglement monotone,
and both of these are faithful,
$C_{\pinv{\vs{\xi}}}(\varrho)=0 \inliff \varrho\in\mathcal{D}_{\pinv{\vs{\xi}}\text{-unc}}$,
$E_{\pinv{\vs{\xi}}}(\varrho)=0 \inliff \varrho\in\mathcal{D}_{\pinv{\vs{\xi}}\text{-sep}}$,
moreover, they show the same lattice structure as the integer partition ideals, $\pinv{P}_\text{II}$,
that is,
\begin{subequations}
\label{eq:mmIIp}
\begin{align}
\label{eq:CmmIIp}
\pinv{\vs{\upsilon}}\finereq\pinv{\vs{\xi}} &\dspiff C_{\pinv{\vs{\upsilon}}}\geq C_{\pinv{\vs{\xi}}},\\
\label{eq:EmmIIp}
\pinv{\vs{\upsilon}}\finereq\pinv{\vs{\xi}} &\dspiff E_{\pinv{\vs{\upsilon}}}\geq E_{\pinv{\vs{\xi}}},
\end{align}
\end{subequations}
which we call \emph{multipartite monotonicity} for Level~II of the permutation invariant case.
(These come by using \eqref{eq:embedPpIIPII} and \eqref{eq:labelsPpII} in Appendix~\ref{app:structPI.part}, 
then \eqref{eq:mmII} and \eqref{eq:measIIp}.)

\subsection{Level~III: classes}
\label{sec:perminv.LIIIp}

Let the map $s$ in \eqref{eq:sII} act elementwisely on the partition ideal filters $\vvs{\xi}\in P_\text{III}$, 
\begin{equation}
\label{eq:sIII}
s(\vvs{\xi}) := \bigsset{s(\vs{\xi})}{\forall \vs{\xi}\in \vvs{\xi}},
\end{equation}
where on the right-hand side a \emph{set} stays.
We call $s(\vvs{\xi})$ the \emph{type} of the set partition ideal filter $\vvs{\xi}$. 
This is a set of integer partition filters, which is denoted as
\begin{subequations}
\begin{equation}
\pinv{\vvs{\xi}} = \bigset{\pinv{\vs{\xi}}_1,\pinv{\vs{\xi}}_2,\dots,\pinv{\vs{\xi}}_{\abs{\pinv{\vvs{\xi}}}}} \subseteq \pinv{P}_\text{II},
\end{equation}
which is a nonempty filter of ideals of integer partitions, the set of which is denoted as
\begin{equation}
\label{eq:PIIIp}
\pinv{P}_\text{III} := s(P_\text{III}) 
=\mathcal{O}_\uparrow(\pinv{P}_\text{II})\setminus\set{\emptyset}.
\end{equation}
This might not seem to be obvious, this is given in \eqref{eq:cd2}, proven in Corollary \ref{cor:commut} in Appendix~\ref{app:structPI.part}.
Based on the refinement of the set partition ideal filters \eqref{eq:poIII}, 
we define a partial order over the integer partition ideal filters,
which we call \emph{refinement} again, and denote with $\preceq$ again, 
given as
\begin{equation}
\label{eq:poIIIp}
\begin{split}
\pinv{\vvs{\upsilon}}\finereq\pinv{\vvs{\xi}} &\dspdef 
\exists \vvs{\upsilon}\in s^{-1}(\pinv{\vvs{\upsilon}}),\vvs{\xi}\in s^{-1}(\pinv{\vvs{\xi}}) \;\text{s.t.}\; \vvs{\upsilon}\finereq\vvs{\xi}\\
&\dspiff \pinv{\vvs{\upsilon}}\subseteq \pinv{\vvs{\xi}},
\end{split}
\end{equation}
\end{subequations}
which turns out to be the inclusion, being the natural partial order for filters.
This might not seem to be obvious, this is given in \eqref{eq:cd2}, proven in Corollary \ref{cor:commut} in Appendix~\ref{app:structPI.part}.
(For illustrations, see Figure~\ref{fig:Pps2345}.)
Note that here, contrary to Level~I and II, we call $\pinv{\vvs{\xi}}$ finer and $\pinv{\vvs{\upsilon}}$ coarser.
Because of these, $\pinv{P}_\text{III}$ with this partial order \emph{is} a lattice.
By construction,
$s$ is monotone with respect to these partial orders,
$\vvs{\upsilon}\preceq\vvs{\xi} \inlthen s(\vvs{\upsilon})\preceq s(\vvs{\xi})$
(see \eqref{eq:mon} in Appendix~\ref{app:structPI.coars}).

With respect to the integer partition ideal filters $\pinv{\vvs{\xi}}\in \pinv{P}_\text{III}$,
we can define the strict partial correlation and entanglement properties.

The \emph{strictly $\pinv{\vvs{\xi}}$-uncorrelated states} are those which are
uncorrelated with respect to all $\pinv{\vs{\xi}} \in\pinv{\vvs{\xi}}$,
and correlated with respect to all $\pinv{\vs{\xi}}\in\cmpl{\pinv{\vvs{\xi}}}=\pinv{P}_\text{II}\setminus\pinv{\vvs{\xi}}$,
so the \emph{class} of these 
(\emph{permutation invariant partial correlation class}) is
\begin{subequations}
\begin{equation}
\label{eq:CuncIIIp}
\mathcal{C}_{\pinv{\vvs{\xi}}\text{-unc}} :=
 \bigcap_{\pinv{\vs{\xi}}\in\cmpl{\pinv{\vvs{\xi}}}} \cmpl{\mathcal{D}_{\pinv{\vs{\xi}}\text{-unc}}} \cap 
 \bigcap_{\pinv{\vs{\xi}}\in      \pinv{\vvs{\xi}} }       \mathcal{D}_{\pinv{\vs{\xi}}\text{-unc}}.
\end{equation}
The \emph{strictly $\pinv{\vvs{\xi}}$-separable states} are those which are
separable with respect to all $\pinv{\vs{\xi}} \in\pinv{\vvs{\xi}}$,
and entangled with respect to all $\pinv{\vs{\xi}}\in\cmpl{\pinv{\vvs{\xi}}}$,
so the \emph{class} of these 
(\emph{permutation invariant partial separability class}, or \emph{permutation invariant partial entanglement class}) is
\begin{equation}
\label{eq:CsepIIIp}
\mathcal{C}_{\pinv{\vvs{\xi}}\text{-sep}} :=
 \bigcap_{\pinv{\vs{\xi}}\in\cmpl{\pinv{\vvs{\xi}}}} \cmpl{\mathcal{D}_{\pinv{\vs{\xi}}\text{-sep}}} \cap 
 \bigcap_{\pinv{\vs{\xi}}\in      \pinv{\vvs{\xi}} }       \mathcal{D}_{\pinv{\vs{\xi}}\text{-sep}}.
\end{equation}
\end{subequations}

The meaning of the permutation invariant Level~III hierarchy can also be clarified.
If there exists a $\varrho \in \mathcal{C}_{\pinv{\vvs{\upsilon}}\text{-unc}}$ and an LO
mapping it into $\mathcal{C}_{\pinv{\vvs{\xi}}\text{-unc}}$, then $\pinv{\vvs{\upsilon}}\finereq\pinv{\vvs{\xi}}$;
and
if there exists a $\varrho \in \mathcal{C}_{\pinv{\vvs{\upsilon}}\text{-sep}}$ and an LOCC
mapping it into $\mathcal{C}_{\pinv{\vvs{\xi}}\text{-sep}}$, then $\pinv{\vvs{\upsilon}}\finereq\pinv{\vvs{\xi}}$.
In this sense, the permutation invariant Level~III hierarchy compares the strength of
correlation and entanglement among the permutation invariant classes.
(These come by the analogue result for the general case in Section~\ref{sec:general.LIII},
for the coarsened classification based on the permutation invariant properties
$P_\text{II*}=\sset{\vee s^{-1}(\pinv{\vs{\xi}})}{\pinv{\vs{\xi}}\in\pinv{P}_\text{II}}$,
see in the next section.)

\begin{figure}\centering
\includegraphics{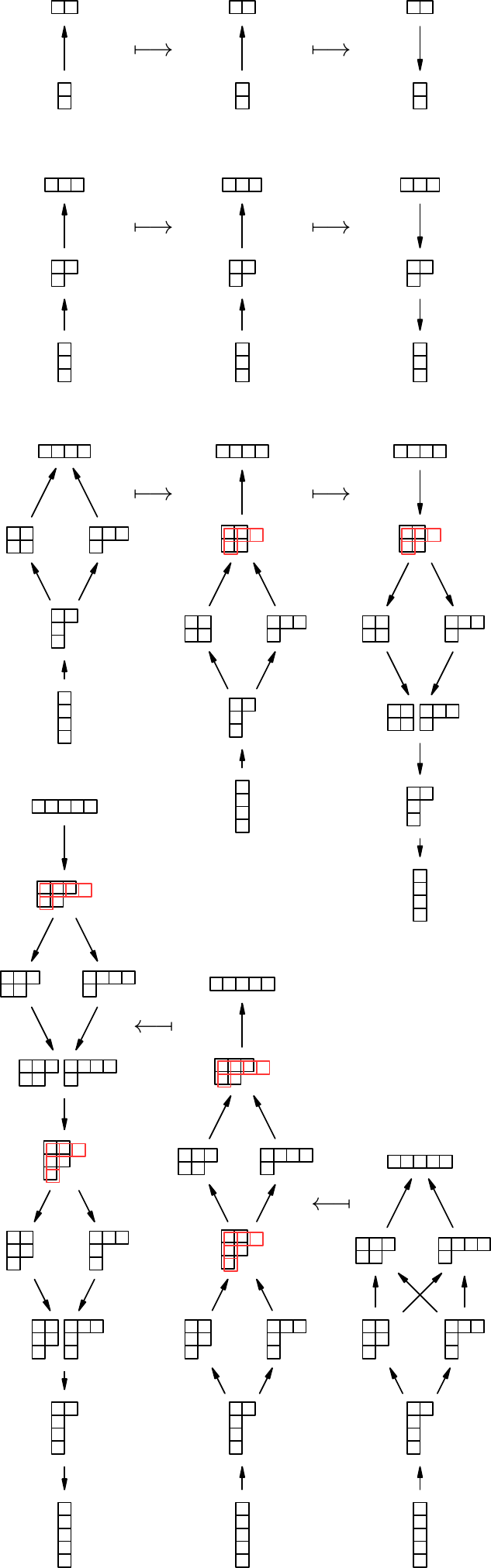}
\caption{The posets of the three-level structure of permutation invariant multipartite correlation and entanglement for $n=2,3,4$ and $5$.
Only the maximal elements of the down-sets of $\pinv{P}_\text{I}$ are shown (with different colors) in $\pinv{P}_\text{II}$, while
only the minimal elements of the up-sets of $\pinv{P}_\text{II}$ are shown (side by side) in $\pinv{P}_\text{III}$.
The partial orders \eqref{eq:poIp}, \eqref{eq:poIIp} and \eqref{eq:poIIIp} are represented by consecutive arrows.}\label{fig:Pps2345}
\end{figure}

\section{An alternative way of introducing permutation invariance}
\label{sec:altperminv}

The Level~II structure $P_\text{II}$ 
encodes the different partial correlation and entanglement properties.
In Sections~\ref{sec:general} and \ref{sec:perminv},
we have built up the structure of the classification,
parallel for the general, and for the permutation invariant cases.
Here we show a more compact, but less transparent treatment of the same structure,
by embedding both constructions into $P_\text{II}$.

The three-level building of the correlation and entanglement classification
was given by the construction
\begin{equation*}
P_\text{I} \;\;\mapsto\;\;
P_\text{II} =\mathcal{O}_\downarrow(P_\text{I})\setminus\set{\emptyset} \;\;\mapsto\;\;
P_\text{III} =\mathcal{O}_\uparrow(P_\text{II})\setminus\set{\emptyset}
\end{equation*}
in Section~\ref{sec:general}.
Because of \eqref{eq:oisomDI},
one may embed $P_\text{I}$ into $P_\text{II}$,
using the \emph{principal ideals} in $P_\text{II}$,
\begin{subequations}
\begin{equation}
P_\text{I}\cong P_\text{II pr.} := \bigsset{\downset\set{\xi}}{\xi\in P_\text{I}}\subseteq P_\text{II},
\end{equation}
by noting that
\begin{equation}
\upsilon \preceq \xi \dspiff \downset\set{\upsilon} \preceq \downset\set{\xi}.
\end{equation}
\end{subequations}
In this way, $P_\text{I}$ is a sublattice of $P_\text{II}$.

The same can be done for the permutation invariant correlation and entanglement classification
given by the construction
\begin{equation*}
\pinv{P}_\text{I} \;\;\mapsto\;\;
\pinv{P}_\text{II} =\mathcal{O}_\downarrow(\pinv{P}_\text{I})\setminus\set{\emptyset} \;\;\mapsto\;\;
\pinv{P}_\text{III} =\mathcal{O}_\uparrow(\pinv{P}_\text{II})\setminus\set{\emptyset}
\end{equation*}
in Section~\ref{sec:perminv}.
In the same way, because of \eqref{eq:oisomDIp},
one may embed $\pinv{P}_\text{I}$ into $\pinv{P}_\text{II}$,
using the \emph{principal ideals} in $\pinv{P}_\text{II}$,
\begin{subequations}
\begin{equation}
\pinv{P}_\text{I}\cong \pinv{P}_\text{II pr.} := \bigsset{\downset\set{\pinv{\xi}}}{\pinv{\xi}\in \pinv{P}_\text{I}}\subseteq \pinv{P}_\text{II},
\end{equation}
by noting that
\begin{equation}
\pinv{\upsilon} \preceq \pinv{\xi} \dspiff \downset\set{\pinv{\upsilon}} \preceq \downset\set{\pinv{\xi}}.
\end{equation}
\end{subequations}
In this way, $\pinv{P}_\text{I}$ is a sublattice of $\pinv{P}_\text{II}$.

The third point here is that,
noting that $\downset s^{-1}(\pinv{\xi}) = \vee s^{-1}(\downset\set{\pinv{\xi}})$
(for the proof, see \eqref{eq:labelsPpIpr} in Appendix~\ref{app:structPI.part}),
one may also embed $\pinv{P}_\text{II}$ (thus also $\pinv{P}_\text{I}$) into $P_\text{II}$,
using
\begin{subequations}
\begin{align}
\pinv{P}_\text{I}  &\cong                     
\bigsset{\vee s^{-1}(\downset\set{\pinv{\xi}})}{\pinv{\xi}\in \pinv{P}_\text{I}}
\subset P_\text{II},\\
\pinv{P}_\text{II} &\cong 
\bigsset{\vee s^{-1}(\pinv{\vs{\xi}})}{\pinv{\vs{\xi}}\in \pinv{P}_\text{II}}
=:P_\text{II pinv.} \subset P_\text{II},
\end{align}
by noting that
\begin{align}
\pinv{\upsilon} \preceq \pinv{\xi} &\dspiff 
  \vee s^{-1}(\downset\set{\pinv{\upsilon}}) \preceq  \vee s^{-1}(\downset\set{\pinv{\xi}}),\\
\pinv{\vs{\upsilon}} \preceq \pinv{\vs{\xi}} &\dspiff 
  \vee s^{-1}(\pinv{\vs{\upsilon}}) \preceq \vee s^{-1}(\pinv{\vs{\xi}}).
\end{align}
\end{subequations}
(For the proof, see \eqref{eq:embedPpIPII}
and \eqref{eq:embedPpIIPII} in Appendix~\ref{app:structPI.part}.)
In this way, $\pinv{P}_\text{I}$ and $\pinv{P}_\text{II}$ are sublattices of $P_\text{II}$.

A $\vs{\xi}\in P_\text{II}$ is permutation invariant, if and only if $\vee s^{-1}(s(\vs{\xi}))=\vs{\xi}$.
(Then $\vs{\xi}\in P_\text{II}$ describes the same property as $\pinv{\vs{\xi}}:=s(\vs{\xi})\in\pinv{P}_\text{II}$.)
So $P_\text{II pinv.}$ can be given directly as
\begin{equation}
P_\text{II pinv.} = \bigsset{\vs{\xi}\in\pinv{P}_\text{II}}{\vee s^{-1}(s(\vs{\xi}))=\vs{\xi}},
\end{equation}
and the permutation invariant classification can be described as a coarsened classification
$\pinv{P}_\text{III}\cong P_\text{III*}=\mathcal{O}_\uparrow(P_\text{II*})\setminus\set{\emptyset}$
with respect to the permutation invariant properties $P_\text{II*}=P_\text{II pinv.}$.

We note that $P_\text{II}$ could also be embedded into $P_\text{III}$ by principal filters
(and similarly $\pinv{P}_\text{II}$ into $\pinv{P}_\text{III}$),
however, the construction has led to a different direction.

\section{\texorpdfstring{$k$}{k}-partitionability, \texorpdfstring{$k$}{k}-producibility and \texorpdfstring{$k$}{k}-stretchability}
\label{sec:kpps}

In this section,
we consider $k$-partitionability and $k$-producibility
as particular cases of permutation invariant properties of partial correlation and entanglement,
and elaborate a duality between them.
Our point of view makes possible to reveal a new property, which we call $k$-stretchability,
combining some advantages of $k$-partitionability and $k$-producibility.
We also investigate the relations among these three properties.

\subsection{\texorpdfstring{$k$}{k}-partitionability and \texorpdfstring{$k$}{k}-producibility of correlation and entanglement}
\label{sec:kpps.pprecall}

\begin{figure*}\centering
\includegraphics{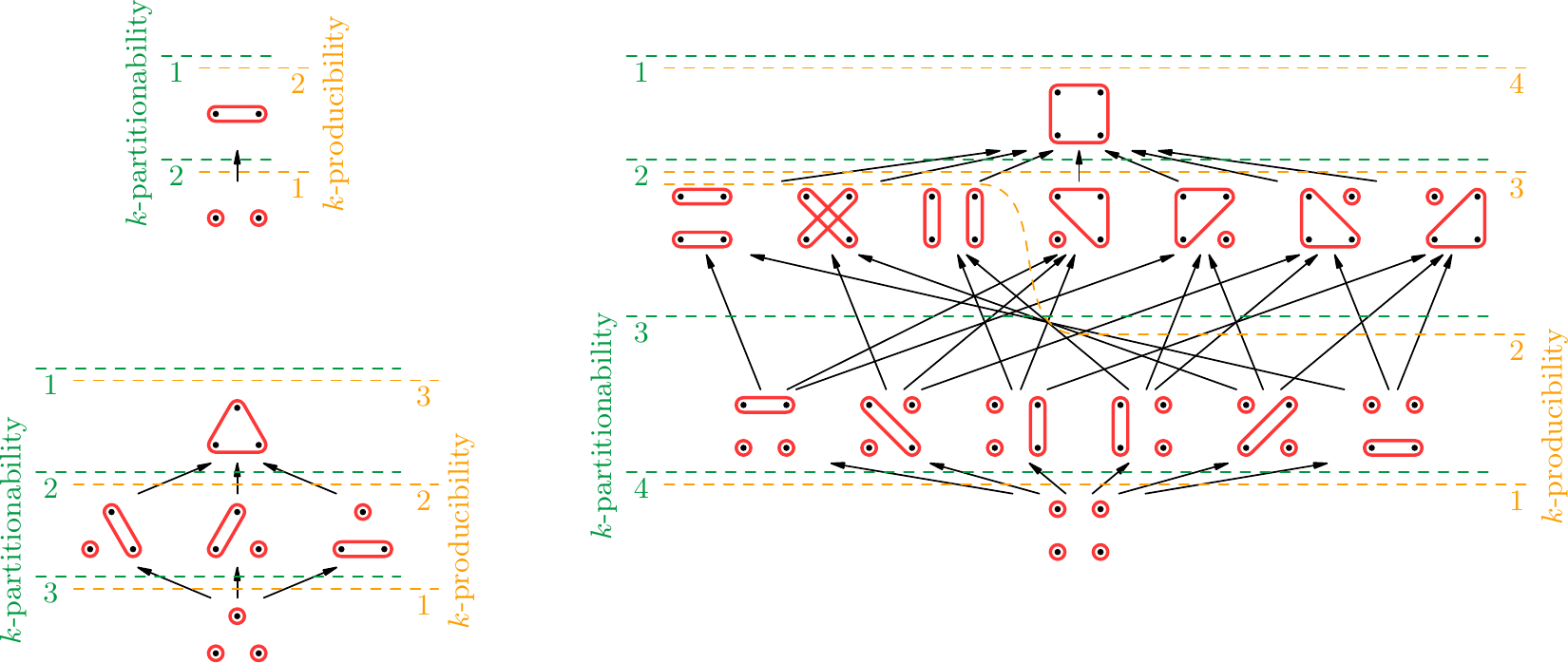}
\caption{The down-sets corresponding to $k$-partitionability and $k$-producibility, 
illustrated on the lattices $P_\text{I}$ for $n=2,3$ and $4$.
(The two kinds of down-sets contain the elements below the specific green and yellow dashed lines.)
}\label{fig:PIpp234}
\end{figure*}

$k$-partitionability and  $k$-producibility are permutation invariant Level~II properties.
A partition is $k$-partitionable, if the number of its parts is at least $k$,
while it is $k$-producible, if the sizes of its parts are at most $k$.
(For illustration, see Figure~\ref{fig:PIpp234}.)
These are encoded by the ideals of \emph{$k$-partitionable} and \emph{$k$-producible partitions},
\begin{subequations}
\label{eq:xipartprod}
\begin{align}
\label{eq:xipart}
\vs{\xi}_{k\text{-part}} &:= \bigsset{\xi\in P_\text{I}}{\abs{\xi}\geq k} \in P_\text{II},\\
\label{eq:xiprod}
\vs{\xi}_{k\text{-prod}} &:= \bigsset{\xi\in P_\text{I}}{\forall X\in\xi: \abs{X}\leq k} \in P_\text{II},
\end{align}
\end{subequations}
for $k = 1,2,\dots,n$,
forming chains in the lattice $P_\text{II}$,
\begin{subequations}
\label{eq:xipartprodch}
\begin{align}
\label{eq:xipartch}
\vs{\xi}_{l\text{-part}} \preceq \vs{\xi}_{k\text{-part}}  &\dspiff l \geq k,  \\
\label{eq:xiprodch}
\vs{\xi}_{l\text{-prod}} \preceq \vs{\xi}_{k\text{-prod}}  &\dspiff l\leq k.
\end{align}
\end{subequations}
Since $k$-partitionability and  $k$-producibility are permutation invariant properties,
it is enough to consider their types \eqref{eq:sII}
\begin{subequations}
\label{eq:xipartprodp}
\begin{align}
\label{eq:xipartp}
\begin{split}
\pinv{\vs{\xi}}_{k\text{-part}} &:= s(\vs{\xi}_{k\text{-part}}) \\
&\phantom{:}= \bigsset{\pinv{\xi}\in \pinv{P}_\text{I}}{\abs{\pinv{\xi}}\geq k} \in \pinv{P}_\text{II},
\end{split}\\
\label{eq:xiprodp}
\begin{split}
\pinv{\vs{\xi}}_{k\text{-prod}} &:= s(\vs{\xi}_{k\text{-prod}}) \\
&\phantom{:}= \bigsset{\pinv{\xi}\in \pinv{P}_\text{I}}{\forall x\in \pinv{\xi}: x\leq k} \in \pinv{P}_\text{II},
\end{split}
\end{align}
\end{subequations}
forming chains in the lattice $\pinv{P}_\text{II}$,
\begin{subequations}
\label{eq:xipartprodpch}
\begin{align}
\label{eq:xipartpch}
\pinv{\vs{\xi}}_{l\text{-part}} \preceq \pinv{\vs{\xi}}_{k\text{-part}}  &\dspiff l\geq k,  \\
\label{eq:xiprodpch}
\pinv{\vs{\xi}}_{l\text{-part}} \preceq \pinv{\vs{\xi}}_{k\text{-prod}}  &\dspiff l\leq k.
\end{align}
\end{subequations}

The corresponding
\emph{$k$-partitionably uncorrelated} and
\emph{$k$-producibly uncorrelated states} \eqref{eq:DuncIIp} are
\begin{subequations}
\label{eq:Dkpartproduncsep}
\begin{align}
\label{eq:Dkpartunc}
\begin{split}
&\mathcal{D}_{k\text{-part unc}}
 := \mathcal{D}_{\vs{\xi}_{k\text{-part}}\text{-unc}}
  = \mathcal{D}_{\pinv{\vs{\xi}}_{k\text{-part}}\text{-unc}}\\
 &= \Bigsset{\bigotimes_{X\in\xi}\varrho_X}{\forall \xi\in P_\text{I} \dispt{s.t.} \abs{\xi}\geq k },
\end{split}\\
\label{eq:Dkprodunc}
\begin{split}
&\mathcal{D}_{k\text{-prod unc}}
:= \mathcal{D}_{\vs{\xi}_{k\text{-prod}}\text{-unc}}
  = \mathcal{D}_{\pinv{\vs{\xi}}_{k\text{-prod}}\text{-unc}}\\
 &= \Bigsset{\bigotimes_{X\in\xi}\varrho_X}{\forall \xi\in P_\text{I} \dispt{s.t.} \forall X\in \xi: \abs{X}\leq k},
\end{split}
\end{align}
(with $\varrho_X\in\mathcal{D}_X$),
which are products
of density operators of at least $k$ subsystems,
and of density operators of subsystems containing at most $k$ elementary subsystems, respectively.
The corresponding \emph{$k$-partitionably separable} (also called \emph{$k$-separable}~\cite{Acin-2001,Guhne-2005,Seevinck-2008}) and
\emph{$k$-producibly separable} (also called \emph{$k$-producible}~\cite{Seevinck-2001,Guhne-2005,Toth-2010}) states \eqref{eq:DsepIIp} are
\begin{align}
\label{eq:Dkpartsep}
\begin{split}
&\mathcal{D}_{k\text{-part sep}}
:= \mathcal{D}_{\vs{\xi}_{k\text{-part}}\text{-sep}}
  = \mathcal{D}_{\pinv{\vs{\xi}}_{k\text{-part}}\text{-sep}}\\
 &= \Conv \mathcal{D}_{k\text{-part unc}}\\
 &= \Bigsset{\sum_j p_j \bigotimes_{X\in\xi_j}\varrho_{X,j}}{\forall j, \forall \xi_j\in P_\text{I} \dispt{s.t.} \abs{\xi_j}\geq k },
\end{split}\\
\label{eq:Dkprodsep}
\begin{split}
&\mathcal{D}_{k\text{-prod sep}}
:=\mathcal{D}_{\vs{\xi}_{k\text{-prod}}\text{-sep}}
  = \mathcal{D}_{\pinv{\vs{\xi}}_{k\text{-prod}}\text{-sep}}\\
 &= \Conv \mathcal{D}_{k\text{-prod unc}}\\
 &= \Bigsset{\sum_j p_j \bigotimes_{X\in\xi_j}\varrho_{X,j}}{\forall j, \\ & \qquad\qquad \forall \xi_j\in P_\text{I} \dispt{s.t.}\forall X\in \xi_j: \abs{X}\leq k},
\end{split}
\end{align}
\end{subequations}
which can be decomposed into $k$-partitionably, or $k$-producibly uncorrelated states, respectively.
Because of \eqref{eq:oisomDuncIIp},
these properties show the same chain structure 
as the chains of the corresponding partition ideals \eqref{eq:xipartprodpch},
that is,
\begin{subequations}
\label{eq:oisomDkpartproduncsep}
\begin{align}
\label{eq:oisomDkpartunc}
l  \geq k   &\dspiff  \mathcal{D}_{l\text{-part unc}} \subseteq \mathcal{D}_{k\text{-part unc}},\\
\label{eq:oisomDkprodunc}
l \leq k  &\dspiff  \mathcal{D}_{l\text{-prod unc}} \subseteq \mathcal{D}_{k\text{-prod unc}},\\
\label{eq:oisomDkpartsep}
l  \geq k   &\dspiff  \mathcal{D}_{l\text{-part sep}} \subseteq \mathcal{D}_{k\text{-part sep}},\\
\label{eq:oisomDkprodsep}
l \leq k  &\dspiff  \mathcal{D}_{l\text{-prod sep}} \subseteq \mathcal{D}_{k\text{-prod sep}},
\end{align}
\end{subequations}
that is,
if a state is $l$-partitionably uncorrelated (separable)
then it is also $k$-partitionably uncorrelated (separable) for all $l\geq k$,
and 
if a state is $l$-producibly uncorrelated (separable)
then it is also $k$-producibly uncorrelated (separable) for all $l\leq k$.

The corresponding
\emph{$k$-partitionability correlation} and
\emph{$k$-producibility correlation} \eqref{eq:CIIp} are \cite{Szalay-2017}
\begin{subequations}
\label{eq:CEkpartprod}
\begin{align}
\label{eq:Ckpart}
C_{k\text{-part}} &:= C_{\vs{\xi}_{k\text{-part}}} = C_{\pinv{\vs{\xi}}_{k\text{-part}}} = \min_{\xi:\abs{\xi}\geq k} C_\xi,\\
\label{eq:Ckprod}
C_{k\text{-prod}} &:= C_{\vs{\xi}_{k\text{-prod}}} = C_{\pinv{\vs{\xi}}_{k\text{-prod}}} = \min_{\xi:\forall X\in\xi:\abs{X}\leq k} C_\xi.
\end{align}
The corresponding
\emph{$k$-partitionability entanglement} and
\emph{$k$-producibility entanglement} \eqref{eq:EIIp} are \cite{Szalay-2015b}
\begin{align}
\label{eq:Ekpart}
E_{k\text{-part}} &:= E_{\vs{\xi}_{k\text{-part}}}= E_{\pinv{\vs{\xi}}_{k\text{-part}}},\\
\label{eq:Ekprod}
E_{k\text{-prod}} &:= E_{\vs{\xi}_{k\text{-prod}}}= E_{\pinv{\vs{\xi}}_{k\text{-prod}}}.
\end{align}
\end{subequations}
Because of \eqref{eq:mmIIp},
these measures show the same chain structure 
as the chains of the corresponding partition ideals \eqref{eq:xipartprodpch},
that is,
\begin{subequations}
\begin{align}
\label{eq:Cmmkpart}
l \geq k  &\dspiff  C_{l\text{-part}} \geq C_{k\text{-part}},\\
\label{eq:Cmmkprod}
l \leq k  &\dspiff  C_{l\text{-prod}} \geq C_{k\text{-prod}},\\
\label{eq:Emmkpart}
l \geq k  &\dspiff  E_{l\text{-part}} \geq E_{k\text{-part}},\\
\label{eq:Emmkprod}
l \leq k  &\dspiff  E_{l\text{-prod}} \geq E_{k\text{-prod}},
\end{align}
\end{subequations}
that is,
$l$-partitionability correlation (entanglement) is always stronger than
$k$-partitionability correlation (entanglement) for all $l\geq k$, 
and
$l$-producibility    correlation (entanglement) is always stronger than
$k$-producibility    correlation (entanglement) for all $l\leq k$, 
which is the multipartite monotonicity of these measures.

For the labeling of the strict partitionability and producibility properties,
we have the ideal filters
\begin{subequations}
\begin{align}
\pinv{\vvs{\xi}}_{k\text{-part}} &:= \upset\bigset{\pinv{\vs{\xi}}_{k\text{-part}}} = \bigsset{\pinv{\vs{\xi}}_{l\text{-part}}}{l\leq k},\\
\pinv{\vvs{\xi}}_{k\text{-prod}} &:= \upset\bigset{\pinv{\vs{\xi}}_{k\text{-prod}}} = \bigsset{\pinv{\vs{\xi}}_{l\text{-prod}}}{l\geq k},
\end{align}
\end{subequations}
by \eqref{eq:xipartprod} and \eqref{eq:xipartprodch}.
The corresponding 
\emph{strictly $k$-partitionably uncorrelated} and
\emph{strictly $k$-producibly uncorrelated states} \eqref{eq:CuncIIIp} are
\begin{subequations}
\begin{align}
\begin{split}
\mathcal{C}_{k\text{-part unc}}
&:= \mathcal{C}_{\pinv{\vvs{\xi}}_{k\text{-part}}\text{-unc}}\\
&\phantom{:}= \mathcal{D}_{k\text{-part unc}} \setminus \mathcal{D}_{(k+1)\text{-part unc}},
\end{split}\\
\begin{split}
\mathcal{C}_{k\text{-prod unc}}
&:= \mathcal{C}_{\pinv{\vvs{\xi}}_{k\text{-prod}}\text{-unc}}\\
&\phantom{:}= \mathcal{D}_{k\text{-prod unc}} \setminus \mathcal{D}_{(k-1)\text{-prod unc}},
\end{split}
\end{align}
which are products
of nonproduct density operators of exactly $k$ subsystems,
and of nonproduct density operators of subsystems containing exactly $k$ elementary subsystems, respectively.
The corresponding 
\emph{strictly $k$-partitionably separable} and
\emph{strictly $k$-producibly separable states} \eqref{eq:CsepIIIp} are
\begin{align}
\begin{split}
\mathcal{C}_{k\text{-part sep}}
&:= \mathcal{C}_{\pinv{\vvs{\xi}}_{k\text{-part}}\text{-sep}}\\
&\phantom{:}= \mathcal{D}_{k\text{-part sep}} \setminus \mathcal{D}_{(k+1)\text{-part sep}},
\end{split}\\
\begin{split}
\mathcal{C}_{k\text{-prod sep}}
&:= \mathcal{C}_{\pinv{\vvs{\xi}}_{k\text{-part}}\text{-sep}}\\
&\phantom{:}= \mathcal{D}_{k\text{-prod sep}} \setminus \mathcal{D}_{(k-1)\text{-prod sep}},
\end{split}
\end{align}
\end{subequations}
which can be decomposed into $k$-partitionably but not $(k+1)$-partitionably,
and $k$-producibly but not $(k-1)$-producibly uncorrelated states, respectively.
($\mathcal{D}_{(n+1)\text{-part unc/sep}}=\emptyset$ and $\mathcal{D}_{{0}\text{-prod unc/sep}}=\emptyset$ are understood.)

\subsection{Duality by conjugation}
\label{sec:kpps.dual}

Looking at the $k$-partitionability and $k$-producibility properties in Figure~\ref{fig:PIpp234},
and also at their definitions \eqref{eq:xipart} and \eqref{eq:xiprod},
it is not clear how these are related to each other.

In Section~\ref{sec:perminv}, the permutation invariant correlation properties were described by the use of integer partitions,
represented by Young diagrams in the figures.
Important properties of Young diagrams are their \emph{height}, \emph{width} and \emph{rank},
given for an integer partition $\pinv{\xi}\in\pinv{P}_\text{I}$ as
\begin{subequations}
\label{eq:hwr}
\begin{align}
\label{eq:h}
h(\pinv{\xi}) &:= \abs{\pinv{\xi}},\\
\label{eq:w}
w(\pinv{\xi}) &:= \max\pinv{\xi},\\
\label{eq:r}
r(\pinv{\xi}) &:= w(\pinv{\xi}) - h(\pinv{\xi}),
\end{align}
\end{subequations}
and also for set partition $\xi\in P_\text{I}$ as
$h(\xi):=h(s(\xi))=\abs{\xi}$,
$w(\xi):=w(s(\xi))=\max_{X\in\xi}\abs{X}$,
$r(\xi):=r(s(\xi))=w(\xi)-h(\xi)$.
It is enlightening now to arrange the poset $\pinv{P}_\text{I}$ 
according to the height and width of the partitions,
in the way as can be seen in Figure~\ref{fig:PpIpp23456},
i.e., 
the height is increasing downwards,
the width is increasing to the right, then
the rank is increasing up-right.
It is easy to see that 
the height is strictly decreasing, 
the width is increasing,
and the rank is strictly increasing monotone 
with respect to the refinement \eqref{eq:poIp},
\begin{equation}
\label{eq:hwmon}
\begin{split}
&\pinv{\upsilon}\prec\pinv{\xi} \dspthen \\
&\quad h(\pinv{\upsilon}) > h(\pinv{\xi}),\quad
w(\pinv{\upsilon}) \leq w(\pinv{\xi}),\quad
r(\pinv{\upsilon}) < r(\pinv{\xi}).
\end{split}
\end{equation}
(For illustration, see Figure~\ref{fig:PpIpp23456}:
the arrows point always upwards, and possibly to the up-right,
never to the up-left, horizontally or downwards.)

Now the partitionability and producibility properties \eqref{eq:xipart} and \eqref{eq:xiprod}
can be formulated with the height and width as
\begin{subequations}
\label{eq:xipartprodhw}
\begin{align}
\label{eq:xiparth}
\pinv{\vs{\xi}}_{k\text{-part}} &= \bigsset{\pinv{\xi}\in \pinv{P}_\text{I}}{h(\pinv{\xi})\geq k},\\
\label{eq:xiprodw}
\pinv{\vs{\xi}}_{k\text{-prod}} &= \bigsset{\pinv{\xi}\in \pinv{P}_\text{I}}{w(\pinv{\xi})\leq k},
\end{align}
so $k$-partitionability is the minimal height,
$k$-producibility is the maximal width of the partition.
In the height-width based arrangement of $\pinv{P}_\text{I}$ in Figure~\ref{fig:PpIpp23456},
the $k$-partitionable partitions \eqref{eq:xiparth} are below the $k$-th rows (counting from the top),
and $k$-producible    partitions \eqref{eq:xiprodw} are to the left from the $k$-th column (counting from the left).
We can also introduce another property, called \emph{$k$-stretchability}, as
\begin{equation}
\label{eq:xistrr}
\pinv{\vs{\xi}}_{k\text{-str}} = \bigsset{\pinv{\xi}\in \pinv{P}_\text{I}}{r(\pinv{\xi})\leq k},
\end{equation}
for $k = -(n-1),-(n-2),\dots,n-2,n-1$,
so $k$-stretchability is the maximal rank of the partition.
It is easy to see that $\pinv{\vs{\xi}}_{k\text{-str}}\in \pinv{P}_\text{II}$, that is, it is a down-set for all $k$,
by the third inequality in \eqref{eq:hwmon}.
We elaborate on the arising correlation and entanglement properties later, in Section~\ref{sec:kpps.stretch}.
\end{subequations}

\begin{figure*}\centering
\includegraphics{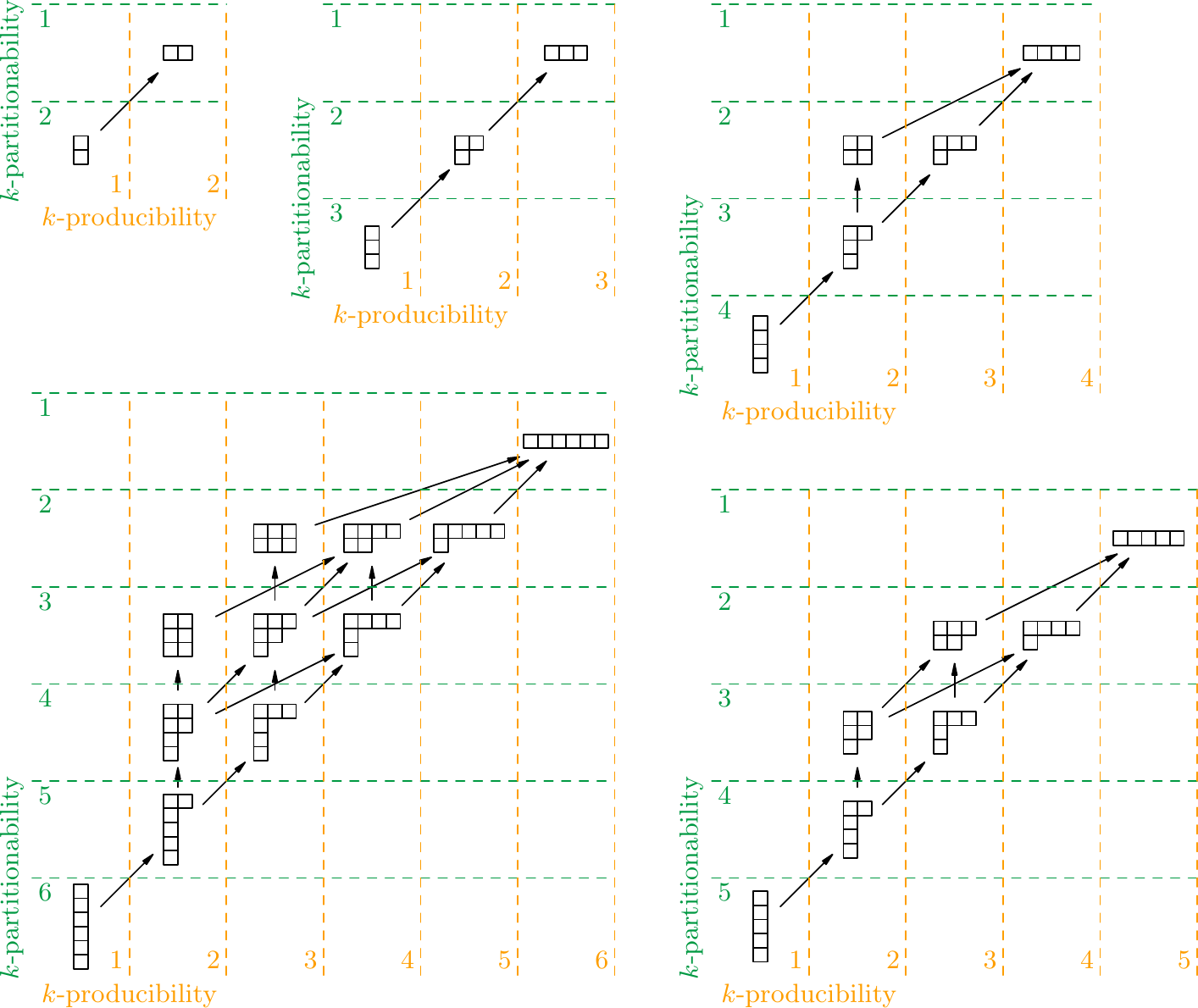}
\caption{The down-sets corresponding to $k$-partitionability and $k$-producibility, 
illustrated on the lattices $\pinv{P}_\text{I}$ for $n=2,3,4,5$ and $6$.
(The two kinds of down-sets contain the elements below the green and to the left from the yellow dashed lines.)
For the first three cases, compare with $P_\text{I}$ in Figure~\ref{fig:PIpp234}.
}\label{fig:PpIpp23456}
\end{figure*}

On the set of Young diagrams,
an involution arises naturally, called \emph{conjugation},
being the reflection of the diagram with respect to its ``diagonal'' \cite{Stanley-2012}.
Considering the integer partitions themselves, the conjugation is given as follows.
For decreasingly ordered values $(x_1,x_2,\dots,x_{\abs{\pinv{\xi}}})$ of the elements of $\pinv{\xi}$,
the number of elements $x'$ that equal to $i$ is $x_i-x_{i+1}$ in the conjugated partition (setting $x_{i+1}=0$ for $i=h(\pinv{\xi})$).
With multisets, the conjugation is given as
the map
$\pinv{P}_\text{I} \to\pinv{P}_\text{I}$,
\begin{equation}
\label{eq:conj}
\pinv{\xi}\mapsto\pinv{\xi}^\pconj
=\bigsmset{  \abs{\smset{x\in\pinv{\xi}}{x\geq i} } }{i=1,2,\dots,w(\pinv{\xi})},
\end{equation}
where both sets on the right-hand side are multisets.

The partial order \eqref{eq:poIp} does not show nice properties with respect to the conjugation.
(For illustration, see the height-width based arrangement in Figure~\ref{fig:PpIpp23456}.)
The conjugation is not an anti-automorphism,
if $\pinv{\upsilon}\preceq\pinv{\xi}$ then it does not follow that 
$\pinv{\upsilon}^\pconj\succeq\pinv{\xi}^\pconj$
(the first example is for $n=4$, where we have $\mset{2,1,1}\preceq\mset{2,2}$
and $\mset{2,1,1}^\pconj=\mset{3,1}\not\succeq\mset{2,2}=\mset{2,2}^\pconj$).
The conjugation is neither an automorphism of course,
if $\pinv{\upsilon}\preceq\pinv{\xi}$ then it does not follow that 
$\pinv{\upsilon}^\pconj\preceq\pinv{\xi}^\pconj$
(the first example is for $n=2$, where we have $\mset{1,1}\preceq\mset{2}$
and $\mset{1,1}^\pconj=\mset{2}\not\preceq\mset{1,1}=\mset{2}^\pconj$).
An integer partition $\pinv{\xi}\in\pinv{P}_\text{I}$
and its conjugate cannot be ordered by refinement
(the first example is for $n=6$, where we have $\mset{2,2,2}$ and $\mset{3,3}$, 
which are conjugates of each other, and cannot be ordered).

The conjugation interchanges the height and width, and multiplies the rank with $-1$,
\begin{equation}
\label{eq:hwrconj}
h(\pinv{\xi}^\pconj) = w(\pinv{\xi}),\quad
w(\pinv{\xi}^\pconj) = h(\pinv{\xi}),\quad
r(\pinv{\xi}^\pconj) = -r(\pinv{\xi}).
\end{equation}
(For illustration, see the height-width based arrangement in Figure~\ref{fig:PpIpp23456}:
the conjugation brings to the position mirrored with respect to the diagonal $w(\pinv{\xi}) = h(\pinv{\xi})$.)
Note that,
although the conjugation interchanges the height and width,
it does not interchange 
the partitionability \eqref{eq:xiparth} and producibility \eqref{eq:xiprodw} properties,
since these are given in different ways by height and width,
namely, by lower and upper bounds.

\subsection{\texorpdfstring{$k$}{k}-stretchability of correlation and entanglement}
\label{sec:kpps.stretch}

In Section~\ref{sec:kpps.dual}, 
we introduced the permutation invariant property $k$-stretchability \eqref{eq:xistrr}
for $k = -(n-1),-(n-2),\dots,n-2,n-1$.
Because of the third inequality in \eqref{eq:hwmon},
these form a chain in the lattice $\pinv{P}_\text{II}$,
\begin{equation}
\label{eq:xistrkch}
\pinv{\vs{\xi}}_{l\text{-str}} \preceq \pinv{\vs{\xi}}_{k\text{-str}}  \dspiff l\leq k.
\end{equation}
(For illustration, see the height-width based arrangement in Figure~\ref{fig:PpIpp23456}.)

The corresponding
\emph{$k$-stretchably uncorrelated states} \eqref{eq:DuncIIp} are
\begin{subequations}
\begin{equation}
\begin{split}
&\mathcal{D}_{k\text{-str unc}}
 := \mathcal{D}_{\pinv{\vs{\xi}}_{k\text{-str}}\text{-unc}}\\
&= \Bigsset{\bigotimes_{X\in\xi}\varrho_X}{\forall \xi\in P_\text{I} \dispt{s.t.} \max_{X\in \xi} \abs{X} - \abs{\xi}\leq k },
\end{split}
\end{equation}
which are products
of density operators 
of as few subsystems as possible, of size as large as possible, of difference upper-bounded by $k$.
In this sense, this is a combination of $k$-partitionability and $k$-producibility.
The corresponding \emph{$k$-stretchably separable states} \eqref{eq:DsepIIp} are
\begin{equation}
\begin{split}
&\mathcal{D}_{k\text{-str sep}}
:= \mathcal{D}_{\pinv{\vs{\xi}}_{k\text{-str}}\text{-sep}}
 = \Conv \mathcal{D}_{k\text{-str unc}}\\
&= \Bigsset{\sum_j p_j \bigotimes_{X\in\xi_j}\varrho_{X,j}}{\forall j, \\ &\qquad\qquad \forall\xi_j\in P_\text{I} \dispt{s.t.} \max_{X\in \xi_j} \abs{X} - \abs{\xi_j}\leq k },
\end{split}
\end{equation}
\end{subequations}
which can be decomposed into $k$-stretchably uncorrelated states.
Because of \eqref{eq:oisomDuncIIp},
these properties show the same chain structure
as the chains of the corresponding partition ideals \eqref{eq:xistrkch},
that is,
\begin{subequations}
\label{eq:oisomDkstruncsep}
\begin{align}
\label{eq:oisomDkstrunc}
l \leq k  &\dspiff  \mathcal{D}_{l\text{-str unc}} \subseteq \mathcal{D}_{k\text{-str unc}},\\
\label{eq:oisomDkstrsep}
l \leq k  &\dspiff  \mathcal{D}_{l\text{-str sep}} \subseteq \mathcal{D}_{k\text{-str sep}},
\end{align}
\end{subequations}
that is,
if a state is $l$-stretchably uncorrelated (separable)
then it is also $k$-stretchably uncorrelated (separable) for all $l\leq k$.

The corresponding
\emph{$k$-stretchability correlation} \eqref{eq:CIIp} is
\begin{subequations}
\label{eq:CEkstr}
\begin{align}
\label{eq:Ckstr}
C_{k\text{-str}} := C_{{\vs{\xi}}_{k\text{-str}}} = C_{\pinv{\vs{\xi}}_{k\text{-str}}} = \min_{\xi: \max\limits_{X\in\xi} \abs{X} - \abs{\xi}\leq k} C_\xi.
\end{align}
The corresponding 
\emph{$k$-stretchability entanglement} \eqref{eq:EIIp} is
\begin{align}
\label{eq:Ekstr}
E_{k\text{-str}} := E_{\vs{\xi}_{k\text{-str}}}= E_{\pinv{\vs{\xi}}_{k\text{-str}}}.
\end{align}
\end{subequations}
Because of \eqref{eq:mmIIp},
these measures show the same chain structure
as the chain of the corresponding partition ideals \eqref{eq:xistrkch},
that is,
\begin{subequations}
\begin{align}
\label{eq:Cmmkstr}
l \leq k  &\dspiff  C_{l\text{-str}} \geq C_{k\text{-str}},\\
\label{eq:Emmkstr}
l \leq k  &\dspiff  E_{l\text{-str}} \geq E_{k\text{-str}},
\end{align}
\end{subequations}
that is,
$l$-stretchability correlation (entanglement) is always stronger than
$k$-stretchability correlation (entanglement) for all $l\leq k$,
which is the multipartite monotonicity of these measures.

For the labeling of the strict stretchability properties,
we have the ideal filters
\begin{equation}
\pinv{\vvs{\xi}}_{k\text{-str}} := \upset\bigset{\pinv{\vs{\xi}}_{k\text{-str}}} = \bigsset{\pinv{\vs{\xi}}_{l\text{-str}}}{l\geq k},
\end{equation}
by \eqref{eq:xistrr} and \eqref{eq:xistrkch}.
The corresponding 
\emph{strictly $k$-stretchably uncorrelated states} \eqref{eq:CuncIIIp} are
\begin{subequations}
\begin{equation}
\mathcal{C}_{k\text{-str unc}}
 := \mathcal{C}_{\pinv{\vvs{\xi}}_{k\text{-str}}\text{-unc}}
 = \mathcal{D}_{k\text{-str unc}} \setminus \mathcal{D}_{(k-1)\text{-str unc}},
\end{equation}
which are products
of nonproduct density operators of as few subsystems as possible, of size as large as possible, of difference exactly $k$.
The corresponding
\emph{strictly $k$-stretchably separable states} \eqref{eq:CsepIIIp} are
\begin{equation}
\mathcal{C}_{k\text{-str sep}}
 := \mathcal{C}_{\pinv{\vvs{\xi}}_{k\text{-str}}\text{-sep}}
 = \mathcal{D}_{k\text{-str sep}} \setminus \mathcal{D}_{(k-1)\text{-str sep}},
\end{equation}
\end{subequations}
which can be decomposed into $k$-stretchably but not $(k-1)$-stretchably uncorrelated states.
($\mathcal{D}_{(-n)\text{-str unc/sep}}=\emptyset$ is understood.)

$k$-partitionability \eqref{eq:xiparth}, $k$-producibility \eqref{eq:xiprodw} and $k$-stretchability \eqref{eq:xistrr}
give, of course, a rather coarsened description of permutation invariant correlation or entanglement properties.
Because of this, they must show some disadvantages.
First, $k$-partitionability \eqref{eq:xiparth} 
takes into account only the number of subsystems uncorrelated with (or separable from) one another,
and does not distinguish between the cases 
when subsystems of roughly equal sizes are uncorrelated (or separable),
and when some elementary subsystems are uncorrelated with (or separable from) the rest of a large system,
although these two cases represent highly different situations form a resource-theoretical point of view.
(For illustration, see the rows in Figure~\ref{fig:PpIpp23456}.)
Second, $k$-producibility \eqref{eq:xiprodw} 
takes into account only the size of the largest subsystem uncorrelated with (or separable from) the other part of the system,
and does not distinguish between the cases
when a large subsystem is uncorrelated with (or separable from) a slightly smaller subsystem, 
or many elementary subsystems,
although these two cases represent again highly different situations form a resource-theoretical point of view.
(For illustration, see the columns in Figure~\ref{fig:PpIpp23456}.)
Third, $k$-stretchability \eqref{eq:xistrr} combines the advantages of the previous two properties in a balanced way,
it takes into account a kind of ``difference'' of them,
and does not distinguish among cases with rather different types.
(For illustration, see the lines parallel to the diagonal in Figure~\ref{fig:PpIpp23456}.)

\subsection{Relations among \texorpdfstring{$k$}{k}-partitionability, \texorpdfstring{$k$}{k}-producibility and \texorpdfstring{$k$}{k}-stretchability}
\label{sec:kpps.relations}

The height $h(\pinv{\xi})$,
width $w(\pinv{\xi})$ and
rank $r(\pinv{\xi})$ \eqref{eq:hwr}
of the integer partitions $\pinv{\xi}\in\pinv{P}_\text{I}$ of $n\in\field{N}$
are bounded by one another as
\begin{subequations}
\label{eq:hwrbound}
\begin{align}
\label{eq:hboundw}
n/w                                 &\leq  h  \leq  n+1-w, \\
\label{eq:wboundh}
n/h                                 &\leq  w  \leq  n+1-h,\\
\label{eq:rboundh}
n/h-h                               &\leq  r  \leq  n+1-2h,\\
\label{eq:rboundw}
-(n+1)+2w                           &\leq  r  \leq  w-n/w,\\
\label{eq:hboundr}
\frac12\bigl(\sqrt{r^2+4n}-r\bigr)  &\leq  h  \leq  \frac12(n+1-r),\\
\label{eq:wboundr}
\frac12\bigl(\sqrt{r^2+4n}+r\bigr)  &\leq  w  \leq  \frac12(n+1+r).
\end{align}
\end{subequations}
The first inequality in \eqref{eq:hboundw} can be seen
by noting that $n\leq hw$, because the right-hand side is 
the area of the smallest rectangle into which the Young diagram of $n$ boxes fits.
The second inequality in \eqref{eq:hboundw} can be seen
by noting that $h+w\leq n+1$, because the left-hand side is 
the half circumference of the smallest rectangle into which the Young diagram of $n$ boxes fits.
The inequalities in \eqref{eq:wboundh} come by the same reasoning, or by conjugation \eqref{eq:hwrconj} in \eqref{eq:hboundw}.
The inequalities in \eqref{eq:rboundh} come by subtracting $h$ from \eqref{eq:wboundh}.
The inequalities in \eqref{eq:rboundw} come by an analogous reasoning, or by conjugation \eqref{eq:hwrconj} in \eqref{eq:rboundh}.
The inequalities in \eqref{eq:hboundr} come by solving the respective inequalities in \eqref{eq:rboundh} for $h$.
The inequalities in \eqref{eq:wboundr} come by an analogous reasoning, or by conjugation \eqref{eq:hwrconj} in \eqref{eq:hboundr}.

The first  inequality in \eqref{eq:rboundw} and the second inequalities in the others in \eqref{eq:hwrbound} are saturated
for the partition $\pinv{\xi}=\mset{m,1,1,\dots,1}$, where the integer $1$ occurs $n-m$ times.
The second inequality in \eqref{eq:rboundw} and the first  inequalities in the others in \eqref{eq:hwrbound} are saturated
for the partition $\pinv{\xi}=\mset{m,m,\dots,m,n-(\lceil n/m\rceil-1) m}$, where the integer $m$ occurs $\lceil n/m\rceil-1$ times.
In the cases where noninteger value stays on one side of an inequality,
saturation is understood for the inequality strengthened by the
ceiling or floor functions
$\lceil  x \rceil  = \min\sset{m\in\mathbb{Z}}{m\geq x}$ or
$\lfloor x \rfloor = \max\sset{m\in\mathbb{Z}}{m\leq x}$;
that is, 
if $q\leq m\in\field{N}$, then also $\lceil  q\rceil \leq m$, and
if $q\geq m\in\field{N}$, then also $\lfloor q\rfloor\geq m$.
(For illustration, see Figure~\ref{fig:PpIpp23456}. 
The second inequalities in \eqref{eq:hboundw} and \eqref{eq:wboundh} express that
the arrangement of Young diagrams in Figure~\ref{fig:PpIpp23456} is ``skew upper triangular'',
while the first inequalities in \eqref{eq:hboundw} and \eqref{eq:wboundh} describe 
the ``hyperbolic'' shape of the upper boundary.
The other inequalities in \eqref{eq:hwrbound} give the boundaries in the rank-height or width-rank plane.)

Since the partitionability is related to lower-bounding the height \eqref{eq:xiparth},
      the producibility    is related to upper-bounding the width  \eqref{eq:xiprodw},
and   the stretchability   is related to their difference \eqref{eq:xistrr},
the bounds in \eqref{eq:hwrbound} lead to the following relations between these properties
\begin{subequations}
\label{eq:xipartprodstrbound}
\begin{align}
\label{eq:partboundprod}
\pinv{\vs{\xi}}_{k\text{-part}} &\finereq \pinv{\vs{\xi}}_{(n+1-k)\text{-prod}}, \\
\label{eq:partboundstr}
\pinv{\vs{\xi}}_{k\text{-part}} &\finereq \pinv{\vs{\xi}}_{(n+1-2k)\text{-str}}, \\
\label{eq:prodboundpart}
\pinv{\vs{\xi}}_{k\text{-prod}} &\finereq \pinv{\vs{\xi}}_{(\lceil n/k\rceil)\text{-part}}, \\
\label{eq:prodboundstr}
\pinv{\vs{\xi}}_{k\text{-prod}} &\finereq \pinv{\vs{\xi}}_{(k-\lceil n/k\rceil)\text{-str}}, \\
\label{eq:strboundpart}
\pinv{\vs{\xi}}_{k\text{-str}} &\finereq \pinv{\vs{\xi}}_{\frac12(\lceil\sqrt{k^2+4n}\rceil-k)\text{-part}}, \\
\label{eq:strboundprod}
\pinv{\vs{\xi}}_{k\text{-str}} &\finereq \pinv{\vs{\xi}}_{\frac12(n+1+k)\text{-prod}}. 
\end{align}
\end{subequations}
The relations in \eqref{eq:partboundprod} and \eqref{eq:partboundstr} can be seen
by using the second inequality in \eqref{eq:wboundh} and the second inequality in \eqref{eq:rboundh}, respectively, with \eqref{eq:xiparth}.
The relations in \eqref{eq:prodboundpart} and \eqref{eq:prodboundstr} can be seen
by using the first inequality in \eqref{eq:hboundw} and the second inequality in \eqref{eq:rboundw}, respectively, with \eqref{eq:xiprodw}.
(For noninteger values, strengthening by the ceiling function was also exploited, as before.)
The relations in \eqref{eq:strboundpart} and \eqref{eq:strboundprod} can be seen
by using the first inequality in \eqref{eq:hboundr} and the second inequality in \eqref{eq:wboundr}, respectively, with \eqref{eq:xistrr}.

From the relations \eqref{eq:xipartprodstrbound},
for the inclusion of the state spaces \eqref{eq:Dkpartproduncsep}, we have
\begin{subequations}
\label{eq:ppsincl}
\begin{align}
\mathcal{D}_{k\text{-part unc}} &\subseteq \mathcal{D}_{(n+1-k)\text{-prod unc}},\\
\mathcal{D}_{k\text{-part unc}} &\subseteq \mathcal{D}_{(n+1-2k)\text{-str unc}},\\
\mathcal{D}_{k\text{-prod unc}} &\subseteq \mathcal{D}_{\lceil n/k\rceil\text{-part unc}},\\
\mathcal{D}_{k\text{-prod unc}} &\subseteq \mathcal{D}_{(k-\lceil n/k\rceil)\text{-str unc}},\\
\mathcal{D}_{k\text{-str unc}}  &\subseteq \mathcal{D}_{\frac12(\lceil\sqrt{k^2+4n}\rceil-k)\text{-part unc}},\\
\mathcal{D}_{k\text{-str unc}}  &\subseteq \mathcal{D}_{\frac12(n+1+k)\text{-prod unc}},\\
\mathcal{D}_{k\text{-part sep}} &\subseteq \mathcal{D}_{(n+1-k)\text{-prod sep}},\\
\mathcal{D}_{k\text{-part sep}} &\subseteq \mathcal{D}_{(n+1-2k)\text{-str sep}},\\
\mathcal{D}_{k\text{-prod sep}} &\subseteq \mathcal{D}_{\lceil n/k\rceil\text{-part sep}},\\
\mathcal{D}_{k\text{-prod sep}} &\subseteq \mathcal{D}_{(k-\lceil n/k\rceil)\text{-str sep}},\\
\mathcal{D}_{k\text{-str sep}}  &\subseteq \mathcal{D}_{\frac12(\lceil\sqrt{k^2+4n}\rceil-k)\text{-part sep}},\\
\mathcal{D}_{k\text{-str sep}}  &\subseteq \mathcal{D}_{\frac12(n+1+k)\text{-prod sep}},
\end{align}
\end{subequations}
by the order isomorphisms \eqref{eq:oisomDIIp};
and
for the bounds of the measures \eqref{eq:CEkpartprod}, we have
\begin{subequations}
\label{eq:CEppsbound}
\begin{align}
\label{eq:Cpboundr}
C_{k\text{-part}} &\geq C_{(n+1-k)\text{-prod}},\\
\label{eq:Cpbounds}
C_{k\text{-part}} &\geq C_{(n+1-2k)\text{-str}},\\
\label{eq:Crboundp}
C_{k\text{-prod}} &\geq C_{\lceil n/k\rceil\text{-part}},\\
\label{eq:Crbounds}
C_{k\text{-prod}} &\geq C_{(k-\lceil n/k\rceil)\text{-str}},\\
\label{eq:Csboundp}
C_{k\text{-str}}  &\geq C_{\frac12(\lceil\sqrt{k^2+4n}\rceil-k)\text{-part}},\\
\label{eq:Csboundr}
C_{k\text{-str}}  &\geq C_{\frac12(n+1+k)\text{-prod}},\\
\label{eq:Epboundr}
E_{k\text{-part}} &\geq E_{(n+1-k)\text{-prod}},\\
\label{eq:Epbounds}
E_{k\text{-part}} &\geq E_{(n+1-2k)\text{-str}},\\
\label{eq:Erboundp}
E_{k\text{-prod}} &\geq E_{\lceil n/k\rceil\text{-part}},\\
\label{eq:Erbounds}
E_{k\text{-prod}} &\geq E_{(k-\lceil n/k\rceil)\text{-str}},\\
\label{eq:Esboundp}
E_{k\text{-str}}  &\geq E_{\frac12(\lceil\sqrt{k^2+4n}\rceil-k)\text{-part}},\\
\label{eq:Esboundr}
E_{k\text{-str}}  &\geq E_{\frac12(n+1+k)\text{-prod}},
\end{align}
\end{subequations}
by the multipartite monotonicity \eqref{eq:mmIIp}.

\section{Summary, remarks and open questions}
\label{sec:summ}

In this work we investigated
the \emph{partial correlation and entanglement properties} which are \emph{invariant under the permutations of the subsystems.}
The set partition based three-level structure,
 describing the classification of partial correlation and entanglement (Section~\ref{sec:general})
was mapped to a parallel, integer partition based three-level structure,
 describing the permutation invariant case (Section~\ref{sec:perminv}).
This mapping is easy to understand on Level~I of the construction,
however, to see that it is working well through the whole construction,
a formal proof was given (Appendix~\ref{app:structPI}).
The construction can be made more compact, although less transparent (Section~\ref{sec:altperminv}),
which can be used as a starting point of more advanced investigations.

We also investigated $k$-partitionability and $k$-producibility,
fitting naturally into the structure of permutation invariant properties.
A kind of combination of these two gives $k$-stretchability,
which is sensitive in a balanced way to
both the maximal size of correlated (or entangled) subsystems
 and the minimal number of subsystems uncorrelated with (or separable from) one another.
We studied their relations, and a duality, connecting the former two (Section~\ref{sec:kpps}).

In the following, we list some remarks and open questions.

The first point to note is that 
we have followed a treatment of entanglement,
which is somewhat different than the standard LOCC paradigm \cite{Werner-1989}.
The LOCC paradigm grabs the essence of entanglement 
as a correlation which cannot be created or increased by classical communication (classical interaction). 
This point of view leads to a classification too detailed and practically unaccomplishable for the multipartite scenario.
Our treatment is rooted more in statistics,
by noticing that
(i) \emph{pure states of classical systems are always uncorrelated,
so in pure states, correlations are of quantum origin,
and this is what we call entanglement~\cite{Schrodinger-1935a};} and
(ii) \emph{mixed states of classical systems can always be formed by forgetting about the identity of pure (hence uncorrelated) states,
so in mixed states, correlations which cannot be described in this way are of quantum origin,
and this is what we call entanglement.}
These principles are working painlessly in the multipartite scenario,
entangled states and correlation based definitions of entanglement measures
were defined based on these in the first two levels of the structure of multipartite entanglement.
In particular, 
from (i) it follows that entanglement in pure states should be measured by correlation 
(see \eqref{eq:EpI} and \eqref{eq:EpII}, and the corresponding quantities for the permutation invariant case,
and the notes in the end of Section~\ref{sec:general.LI}).
We emphasize that this statistical point of view is fully compatible with LOCC,
and has led to a much coarser, finite classification.

Being uncertain (or forgetting) about the identity of the state of the system
is a guiding principle in the definition of the different aspects of multipartite entanglement
in Section~\ref{sec:general}
(see also in point (ii) in Section VII.A. in \cite{Szalay-2015b}).\\
\textit{States in $\mathcal{D}_L$:}
We are uncertain about the (pure) state, by which the system is described
(Section~\ref{sec:general.L0}).\\
\textit{States in $\mathcal{D}_{\xi\text{-unc}}$ ($\mathcal{D}_{\xi\text{-sep}}$):}
We are uncertain about the pure state, by which the system is described,
but we are certain about the partition with respect to which the state is uncorrelated (separable)
(Section~\ref{sec:general.LI}).\\
\textit{States in $\mathcal{D}_{\vs{\xi}\text{-unc}}$ ($\mathcal{D}_{\vs{\xi}\text{-sep}}$):}
We are uncertain about the pure state, by which the system is described,
and we are also uncertain about the partition with respect to which the state is uncorrelated (separable), 
but we are certain about the possible partitions with respect to which the state is uncorrelated (separable)
(Section~\ref{sec:general.LII}).\\
\textit{States in $\mathcal{C}_{\vvs{\xi}\text{-unc}}$ ($\mathcal{C}_{\vvs{\xi}\text{-sep}}$):}
We are uncertain about the pure state, by which the system is described,
and we are also uncertain about the partition with respect to which the state is uncorrelated (separable),
but we are certain about the possible partitions with respect to which the state is uncorrelated (separable),
and we are also certain about the possible partitions with respect to which the state is correlated (entangled)
(Section~\ref{sec:general.LIII}).

We also mention here that
for the state sets of given multipartite correlation (entanglement) properties
$\mathcal{D}_{\vs{\xi}\text{-unc}}$ ($\mathcal{D}_{\vs{\xi}\text{-sep}}$),
except from the fully uncorrelated (or fully separable) case $\vs{\xi}=\set{\bot}$,
we cannot formulate semigroups of quantum channels 
which could play the role of ``free operations'' for these ``free state'' sets
in a usual resource-theoretical scenario \cite{Chitambar-2019}.
Instead of this, being uncertain (or forgetting) about the identity of the \emph{maps} is the guiding principle here
(see the characterizations in Sections~\ref{sec:general.LI} and \ref{sec:general.LII},
as well as in the permutation invariant case in Sections~\ref{sec:perminv.LIp} and \ref{sec:perminv.LIIp}).
Note that this is physical: only the fully uncorrelated (fully separable) states are for free,
the others are resource states, possibly of different ``value.''

Note that, after the general case (Section~\ref{sec:general}),
we considered the
permutation invariant (correlation or entanglement) \emph{properties} (from Section~\ref{sec:perminv}).
The use of these is not restricted to permutation invariant \emph{states}:
these are simply the permutation invariant properties of also non-permutation-invariant states.
If the states considered are permutation invariant
(for example, in the first quantization of bosonic or fermionic systems),
then these are the only relevant properties.

The permutation invariant properties were described by the use of \emph{integer partitions,}
for which a refinement-like order was given,
which is just the coarsening of the refinement order of \emph{set partitions.}
In the literature, 
there are several partial orders constructed for integer partitions,
such as
the \emph{dominance order} \cite{Brylawski-1973}
(based on majorization, leading to a lattice), 
the \emph{Young order} \cite{Andrews-1984,Stanley-2012}
(based on diagram containment,
given for partitions of all integers, leading to \emph{Young's lattice},
here partitions of the same integer cannot be ordered,
and the order is invariant to the conjugation),
or the \emph{reverse lexicographic order} 
(leading to a chain).
The \emph{refinement order} for integer partitions,
introduced in \eqref{eq:poIp} through the \emph{refinement order} for set partitions \eqref{eq:poI},
is different from the above orders,
and, to our knowledge \cite{Andrews-1984,Stanley-2012}, was not considered upon its merits in the literature before.
Although Birkhoff mentioned this partial order in his book \cite{Birkhoff-1973} as an example, 
it was not used for any reasonable purpose.
The reason for this might be that the resulting poset shows properties not so nice or powerful as those shown by the others.
However, this order is what needed is in the classification problem of permutation invariant properties,
that is reflected also by the monotonicities \eqref{eq:hwmon}.

$k$-partitionability is about the natural gradation 
of the lattice of set partitions $P_\text{I}$ or 
of the poset of integer partitions $\pinv{P}_\text{I}$,
but it was not clear, how $k$-producibility can be understood.
The conjugation of integer partitions has explained the role of the latter,
by establishing a kind of duality between the two properties.
Note, however, that this duality is only partial:
although the conjugation of integer partitions \eqref{eq:conj} interchanges the height and the width \eqref{eq:hwrconj},
by which the $k$-partitionability and $k$-producibility properties are given,
but the latter properties are given in an opposite way by height and width, see \eqref{eq:xiparth} and \eqref{eq:xiprodw}.
So the conjugation does not interchange $k$-partitionability and $k$-producibility,
it establishes a connection on the deeper level of integer partitions only.
This is also reflected in the relations among the different $k$-partitionability and $k$-producibility properties,
given in \eqref{eq:partboundprod} and \eqref{eq:prodboundpart}.

For $n\leq6$, the height and width determine $\pinv{\xi}$ uniquely.
For $n=7$, we have $\mset{3,2,2}$ and $\mset{3,3,1}$, being of height $3$ and width $3$.
Note that two different integer partitions of the same height and width can never be ordered.
(Ordered pairs are of different height, 
because of the strict monotonicity of the height \eqref{eq:hwmon}.)

The property of $k$-stretchability \eqref{eq:xistrr} was introduced by the rank of the partition.
We note that the \emph{rank} of an integer partition \eqref{eq:r} was defined and used originally
in a very different context in number theory and combinatorics.
It was introduced by Freeman Dyson \cite{Dyson-1944} 
in his investigations of Ramanujan's congruences in the partition function \cite{oeisA000041}. 

For $n$ subsystems,
the number of nontrivial $k$-partitionability and $k$-producibility properties is $n-1$ in both cases 
($1$-partitionability and $n$-producibility are trivial),
while 
the number of nontrivial $k$-stretchability properties is $2(n-1)$
($(n-1)$-stretchability is trivial).
All of these three properties form chains, see \eqref{eq:xipartpch}, \eqref{eq:xiprodpch} and \eqref{eq:xistrkch},
and the relations among them are given in \eqref{eq:xipartprodstrbound}.
Note that $k$-stretchability combines the advantages of $k$-partitionability and $k$-producibility,
in the sense that 
it rewards large correlated (or entangled) subsystems, 
while it punishes the larger number of subsystems uncorrelated (or separable) with one another.
Also, 
while the relations between $k$-partitionability and $k$-producibility 
are highly nontrivial \eqref{eq:partboundprod}, \eqref{eq:prodboundpart},
the $k$-stretchability properties form a chain \eqref{eq:xistrkch}.
The price to pay for this is that 
sometimes there are rather different properties not distinguished by stretchability,
for example, $\mset{3,3}$ and $\mset{4,1,1}$ are both $1$-stretchable.
Taking stretchability seriously
leads to a highly nontrivial, balanced comparison of permutation invariant multipartite correlation or entanglement properties.

\begin{acknowledgments}
Discussions with \emph{G\'eza T\'oth} and \emph{Mih\'aly M\'at\'e} are gratefully acknowledged.
This research was financially supported by
the {National Research, Development and Innovation Fund of Hungary}
within the \textit{Researcher-initiated Research Program} (project Nr:~NKFIH-K120569)
and within the \textit{Quantum Technology National Excellence Program} (project Nr:~2017-1.2.1-NKP-2017-00001),
the Ministry for Innovation and Technology
within the \textit{{\'U}NKP-19-4 New National Excellence Program},
and the {Hungarian Academy of Sciences}
within the \textit{J\'anos Bolyai Research Scholarship}
and the \textit{``Lend\"ulet'' Program}.
\end{acknowledgments}

\appendix
\section{On the structure of the classification of permutation invariant correlations}
\label{app:structPI}

\subsection{Coarsening a poset}
\label{app:structPI.coars}

Let us have a \emph{finite set} $P$ endowed with a \emph{binary relation} $\preceq$,
and a \emph{function} $f$, mapping $P$ to $P':=f(P)$.
(In the appendix we consider finite sets only, even if it is not mentioned explicitly.)
The same function transforms also the relation naturally as follows.
A binary relation can be considered as a set of ordered pairs in $P$,
written as $\preceq\;\equiv \sset{(b,a)\in P\times P}{b\preceq a}$,
then $\preceq' := (f\times f)(\preceq) = 
\sset{(f(b),f(a))}{ \forall (b,a)\in\;\preceq}$ by the elementwise action of $f\times f$.
We use the shorthand notation $\preceq'=f(\preceq)$ for this.
The meaning of this is that
\begin{equation}
\begin{split}
\label{eq:relp}
&\forall a',b'\in P', \quad b'\preceq'a' \dspiff \\
&\qquad\exists a\in f^{-1}(a'), \exists b\in f^{-1}(b') \dispt{s.t.} b\preceq a.
\end{split}
\end{equation}
(Here $f^{-1}(a')=\sset{a\in P}{f(a)=a'}\subseteq P$ denotes the inverse image of the singleton $\set{a'}$.)

If the binary relation $\preceq$ is a \emph{partial order} \cite{Davey-2002,Stanley-2012},
the transformed one $\preceq'$ is not necessarily that.
To make it a partial order, we need some constraints imposed on $f$.
Let us introduce the following three conditions:
\begin{subequations}
\label{eq:constr}
\begin{align}
\label{eq:densd}
\begin{split}
&\forall a',b'\in P', \dispt{if} b'\preceq'a' \dispt{then} \\
&\qquad\forall a \in f^{-1}(a'), \; \exists b \in f^{-1}(b'), \dispt{s.t.}  b\preceq a;
\end{split} \\
\label{eq:densu}
\begin{split}
&\forall a',b'\in P', \dispt{if} b'\preceq'a' \dispt{then} \\
&\qquad\forall b \in f^{-1}(b'), \; \exists a \in f^{-1}(a'), \dispt{s.t.}  b\preceq a;
\end{split} \\
\label{eq:sol}
\begin{split}
&\forall a,b,c\in P, \dispt{if} c\preceq b\preceq a \\
&\qquad\dispt{and} f(c) = f(a) \dispt{then} f(b)=f(a).
\end{split}
\end{align}
\end{subequations}
These conditions, although being rather strong, hold in the construction in which we need them
(see Appendix~\ref{app:structPI.part}, and Lemma~\ref{lem:ref.constr}). 
They turn out to be sufficient for $\preceq'$ to be a partial order,
as the following lemma states.

\begin{lem}
\label{lem:fpo}
Let $(P,\preceq)$ be a poset, then $(P',\preceq') = (f(P),f(\preceq))$ is a poset
if \eqref{eq:densd} and \eqref{eq:sol} hold, or
if \eqref{eq:densu} and \eqref{eq:sol} hold.
\end{lem}

\begin{proof}
We need to prove that $\preceq'$ is a partial order in these cases.\\
(i) \emph{Reflexivity} ($\forall a'\in P'$, $a'\preceq'a'$):
for all $a'\in P'$ we have $f^{-1}(a')\neq\emptyset$, since $f$ is surjective,
and we need $b\in f^{-1}(a')$ and $a\in f^{-1}(a')$ for which $b\preceq a$ by \eqref{eq:relp};
this holds for the choice $b=a$, since the partial order $\preceq$ is reflexive, $a\preceq a$.
(This proof does not use any of the constraints \eqref{eq:constr}.)\\
(ii) \emph{Antisymmetry} ($\forall a',b'\in P'$, if $a'\preceq'b'$ and $b'\preceq'a'$, then $a'=b'$):
let $b'\preceq'a'$, which means that there exist $a\in f^{-1}(a')$ and $b\in f^{-1}(b')$ such that $b\preceq a$, by \eqref{eq:relp}.
Fix such a pair $b$ and $a$.
Also, let $a'\preceq'b'$, applying \eqref{eq:densd} for $b$, we have that there exists $c\in f^{-1}(a')$, such that $c\preceq b$.
Then, applying \eqref{eq:sol} to $c\preceq b\preceq a$, we have that $b'=f(b)=f(a)=a'$.
(This proof uses \eqref{eq:densd} and \eqref{eq:sol}.
A similar proof can be given by using \eqref{eq:densu} and \eqref{eq:sol}.)\\
(iii) \emph{Transitivity} ($\forall a',b',c'\in P'$, if $c'\preceq'b'$ and $b'\preceq'a'$, then $c'\preceq'a'$):
let $b'\preceq'a'$, which means that there exist $a\in f^{-1}(a')$ and $b\in f^{-1}(b')$ such that $b\preceq a$, by \eqref{eq:relp}.
Fix such a pair $b$ and $a$.
Also, let $c'\preceq'b'$, applying \eqref{eq:densd} for $b$, we have that there exists $c\in f^{-1}(c')$, such that $c\preceq b$.
Then, since the partial order $\preceq$ is transitive, we have that $c\preceq a$ for an $a\in f^{-1}(a')$ and $c\in f^{-1}(c')$,
which means that $c'\preceq'a'$ by \eqref{eq:relp}.
(This proof uses \eqref{eq:densd}.
A similar proof can be given by using \eqref{eq:densu}.)
\end{proof}

The following construction would be much simpler if
the conditions \eqref{eq:constr} would also be necessary.
Note that this is not the case:
if $\preceq'$ is a partial order, then \eqref{eq:sol} holds,
but \eqref{eq:densd} and \eqref{eq:densu} do not hold.
For example, for the poset $P=\set{a_1,a_2,b}$ with the only arrow $b\prec a_1$,
and the function given as $f(a_1)=f(a_2)=a'\neq f(b)=b'$, we have $b'\prec' a'$, but \eqref{eq:densd} does not hold.

If \eqref{eq:densd} or \eqref{eq:densu}, and \eqref{eq:sol} hold, then
the new poset $(P',\preceq')$ can be considered as a coarsening of $(P,\preceq)$ by $f$.
($f:P\to P'$ is surjective by definition.
If it is also injective, then $(P,\preceq)$ and $(P',\preceq')$ are isomorphic,
and the coarsening is trivial.)
Note that, because of the construction, $f:P\to P'$ is automatically monotone,
\begin{equation}
\label{eq:mon}
b\preceq a \dspthen f(b) \preceq' f(a).
\end{equation}
(Indeed, $a\in f^{-1}(f(a))$ and $b\in f^{-1}(f(b))$,
so \eqref{eq:relp} holds.)
Note also that, 
if the poset $(P,\preceq)$ is a lattice,
the transformed one $(P',\preceq')$ is not necessarily that.
To make it a lattice, further conditions have to be imposed; this problem is not addressed here.

\subsection{Down-sets and up-sets}
\label{app:structPI.du}

In a poset $P$,
a \emph{down-set} (\emph{order ideal}) is a subset $\ve{a}\subseteq P$, which is closed downwards \cite{Davey-2002,Stanley-2012},
\begin{subequations}
\begin{equation}
\label{eq:dset}
\ve{a}\in \mathcal{O}_\downarrow(P) \dspdef
\forall a\in\ve{a}, \; \forall b\preceq a: \; b\in\ve{a};
\end{equation}
while an \emph{up-set} (\emph{order filter}) is a subset $\ve{a}\subseteq P$, which is closed upwards \cite{Davey-2002,Stanley-2012},
\begin{equation}
\label{eq:uset}
\ve{a}\in \mathcal{O}_\uparrow(P) \dspdef
\forall a\in\ve{a}, \; \forall b\succeq a: \; b\in\ve{a}.
\end{equation}
\end{subequations}
These form lattices with respect to the inclusion $\subseteq$,
with the join $\vee$ (least upper bound), being the union $\cup$,
and the meet $\wedge$ (greatest lower bound), being the intersection $\cap$.

As in the previous section,
let us have the poset $(P,\preceq)$,
the function $f$,
by which $(P',\preceq') := (f(P),f(\preceq))$,
and
for which \eqref{eq:densd} or \eqref{eq:densu}, and \eqref{eq:sol} hold.
Additionally, let $P$ have \emph{bottom} and \emph{top elements}.
(Then also $P'$ has \emph{bottom} and \emph{top elements}.
Indeed, this is because the bottom and top elements in $P$ are mapped 
to the bottom and top elements in $P'$, because of the \eqref{eq:mon} monotonicity of $f$.)
Now let us form from the posets $(P,\preceq)$ and $(P',\preceq')$
the lattices of nonempty down-sets
\begin{subequations}
\label{eq:buildup}
\begin{align}
(Q ,\sqsubseteq ) &:= (\mathcal{O}_\downarrow(P )\setminus\set{\emptyset},\subseteq),\\
(Q',\sqsubseteq') &:= (\mathcal{O}_\downarrow(P')\setminus\set{\emptyset},\subseteq).
\end{align}
\end{subequations}
(The existence of the bottom elements in $P$ and $P'$
ensures that not only the down-sets but also the nonempty down-sets form lattices in both cases \cite{Szalay-2015b}.)
Then, by denoting the elementwise action of $f$ with $g:2^P\to 2^{P'}$,
that is, $g(\ve{a})=\sset{f(a)}{a\in\ve{a}}$,
in the following lemmas we will show
that $(Q',\sqsubseteq') = (g(Q),g(\sqsubseteq))$,
that is, the diagram 
\begin{equation}
\label{eq:cd}
\xymatrix@M+=8bp{
(Q,\sqsubseteq)   \ar@{|->}[r]^g &
(Q',\sqsubseteq')  \\ 
(P,\preceq)       \ar@{|->}[r]^f \ar@{|->}[u]_{\mathcal{O}_\downarrow\setminus\set{\emptyset}} &
(P',\preceq')     \ar@{|->}[u]_{\mathcal{O}_\downarrow\setminus\set{\emptyset}} 
}
\end{equation}
commutes.
Here $g(\sqsubseteq)$ is the partial order transformed by $g$ in the same way as $\preceq$ was transformed by $f$ in \eqref{eq:relp},
and we also have automatically that $g$ is monotone for $\sqsubseteq$ and $g(\sqsubseteq)$, in the same way as in \eqref{eq:mon}.

The following two technical lemmas concerning the sets $\bigcup_{a'\in\ve{a}'} f^{-1}(a')\subseteq P$
will be used several times later.

\begin{lem}
\label{lem:technical1}
In the above setting
(and assuming \eqref{eq:constr}),
for all $\ve{a}'\in Q'$, we have $g\bigl(\bigcup_{a'\in\ve{a}'} f^{-1}(a')\bigr)=\ve{a}'$.
\end{lem}

\begin{proof}
This is because $g$ is the elementwise action of $f$, so
$g\bigl(\bigcup_{a'\in\ve{a}'} f^{-1}(a')\bigr)
= g\bigl(\bigcup_{a'\in\ve{a}'} \sset{a\in P}{f(a)=a'}\bigr)
= \bigcup_{a'\in\ve{a}'} g\bigl(\sset{a\in P}{f(a)=a'}\bigr)
=\bigcup_{a'\in\ve{a}'} \set{a'} 
= \ve{a}'$.
\end{proof}

\begin{lem}
\label{lem:technical2}
In the above setting
(and assuming \eqref{eq:constr}),
for all $\ve{a}'\in Q'$, we have $\bigcup_{a'\in\ve{a}'} f^{-1}(a')\in Q$,
that is, it is a nonempty down-set.
\end{lem}

\begin{proof}
First, denote $\ve{a}:=\bigcup_{a'\in\ve{a}'} f^{-1}(a')$.
On the one hand, $\ve{a}\neq\emptyset$, since $\ve{a}'\neq\emptyset$ and $f$ is surjective.
On the other hand,
if $b\preceq a$ for an $a\in\ve{a}$,
then $f(b)\preceq'f(a)$ because of the monotonicity \eqref{eq:mon}.
Since $f(a)\in \ve{a}'$ by the construction of $\ve{a}$,
and $\ve{a}'$ is a down-set \eqref{eq:dset},
we have that $f(b)\in\ve{a}'$.
Then $b\in f^{-1}(f(b))\subseteq\ve{a}$ by the construction of 
$\ve{a}= f^{-1}(f(b))\cup \bigcup_{a'\in\ve{a}', a'\neq f(b)} f^{-1}(a')$,
and then $b\in\ve{a}$, so $\ve{a}$ is a down-set.
Altogether we have that $\ve{a}\in Q$.
\end{proof}

The following two lemmas show \eqref{eq:cd}.

\begin{lem}
\label{lem:com.set}
In the above setting
(and assuming \eqref{eq:constr}),
we have $g(Q)=Q'$.
\end{lem}

\begin{proof}
We need to prove both inclusions.\\
(i) We need that $g(Q)\subseteq Q'$.
Let $\ve{a}\in Q$ and $\ve{a}' := g(\ve{a})$.
Then for all $a'\in\ve{a}'$ let us have $a \in f^{-1}(a')\cap \ve{a}$.
(Note that $f^{-1}(a')\cap \ve{a}\neq\emptyset$ by construction.)
Then let $b'\preceq'a'$,
and applying \eqref{eq:densd} for $a$, we have that there exists $b\in f^{-1}(b')$ such that $b\preceq a$.
Since $a\in\ve{a}$, and $\ve{a}$ is a down-set \eqref{eq:dset}, we have that $b\in\ve{a}$,
and then $b'\in\ve{a}'$, so $\ve{a}'$ is a down-set.\\
(ii) We also need that $g(Q)\supseteq Q'$.
Let $\ve{a}' \in Q'$, then we construct an $\ve{a}\in Q$, for which $\ve{a}' = g(\ve{a})$.
The element $\ve{a} := \bigcup_{a'\in\ve{a}'} f^{-1}(a')$ fulfills these criteria,
we have $\ve{a}' = g(\ve{a})$ by Lemma~\ref{lem:technical1},
and $\ve{a}\in Q$ by Lemma~\ref{lem:technical2}.
\end{proof}

\begin{lem}
\label{lem:com.rel}
In the above setting
(and assuming \eqref{eq:constr}),
we have $g(\sqsubseteq)=\;\sqsubseteq'$.
\end{lem}

\begin{proof}
We need to prove both directions.\\
(i) We need that for all $\ve{a}',\ve{b}'\in Q'$,
if $\ve{b}'g(\sqsubseteq)\ve{a}'$ then  $\ve{b}'\sqsubseteq'\ve{a}'$.
$\ve{b}'g(\sqsubseteq)\ve{a}'$ means that 
$\exists\ve{b}\in g^{-1}(\ve{b}')$ and
$\exists\ve{a}\in g^{-1}(\ve{a}')$ such that $\ve{b}\sqsubseteq\ve{a}$.
(Again, $g^{-1}(\ve{a}')=\sset{\ve{a}\in Q}{g(\ve{a})=\ve{a}'}$ is the inverse image of the singleton $\set{\ve{a}'}$.)
In this case,
$\ve{b}' = g(\ve{b}) = \sset{f(b)}{b\in \ve{b}} \subseteq  \sset{f(a)}{a\in \ve{a}} = g(\ve{a}) = \ve{a}'$,
so $\ve{b}'\sqsubseteq'\ve{a}'$.\\
(ii) We also need that for all $\ve{a}',\ve{b}'\in Q'$,
if $\ve{b}'\sqsubseteq'\ve{a}'$ then $\ve{b}'g(\sqsubseteq)\ve{a}'$. 
From $\ve{b}'$ and $\ve{a}'$,
let us form $\ve{b} := \bigcup_{b'\in\ve{b}'} f^{-1}(b')$
and $\ve{a} := \bigcup_{a'\in\ve{a}'} f^{-1}(a')$.
We have that 
$\ve{b}\in Q$ and $\ve{a}\in Q$
(they are nonempty down-sets) by Lemma~\ref{lem:technical2},
also $g(\ve{b})=\ve{b}'$ and $g(\ve{a})=\ve{a}'$
by Lemma~\ref{lem:technical1},
and if $\ve{b}'\sqsubseteq'\ve{a}'$ then $\ve{b}\sqsubseteq\ve{a}$ by construction.
That is, we have constructed 
$\ve{b}\in g^{-1}(\ve{b}')$ and
$\ve{a}\in g^{-1}(\ve{a}')$, for which $\ve{b}\sqsubseteq\ve{a}$,
so we have $\ve{b}'g(\sqsubseteq)\ve{a}'$.
\end{proof}

Summing up, 
Lemma~\ref{lem:com.set} and Lemma~\ref{lem:com.rel} together state that \eqref{eq:cd} commutes.
Later we also need that the properties \eqref{eq:densd}, \eqref{eq:densu} and \eqref{eq:sol} are inherited.

\begin{lem}
\label{lem:constrInherit}
In the above setting
(and assuming \eqref{eq:constr}),
we have 
\begin{subequations}
\label{eq:constr2}
\begin{align}
\label{eq:densd2}
\begin{split}
&\forall \ve{a}',\ve{b}'\in Q', \dispt{if} \ve{b}'\sqsubseteq'\ve{a}' \dispt{then} \\
&\qquad\forall \ve{a} \in g^{-1}(\ve{a}') \; \exists \ve{b} \in g^{-1}(\ve{b}') \dispt{s.t.}  \ve{b}\sqsubseteq \ve{a};
\end{split} \\
\label{eq:densu2}
\begin{split}
&\forall \ve{a}',\ve{b}'\in Q', \dispt{if} \ve{b}'\sqsubseteq'\ve{a}' \dispt{then} \\
&\qquad\forall \ve{b} \in g^{-1}(\ve{b}') \; \exists \ve{a} \in g^{-1}(\ve{a}') \dispt{s.t.}  \ve{b}\sqsubseteq \ve{a};
\end{split} \\
\label{eq:sol2}
\begin{split}
&\forall \ve{a},\ve{b},\ve{c}\in Q, \dispt{if} \ve{c}\sqsubseteq\ve{b}\sqsubseteq\ve{a} \\
&\qquad\dispt{and} g(\ve{c}) = g(\ve{a}), \dispt{then} g(\ve{b})=g(\ve{a}).
\end{split}
\end{align}
\end{subequations}
\end{lem}

\begin{proof}
\eqref{eq:densd2} can be proven by the explicit construction of $\ve{b}$.
Let $\ve{b}'\sqsubseteq'\ve{a}'$, and $\ve{a} \in g^{-1}(\ve{a}')$,
then $\ve{b}:=\bigl(\bigcup_{b'\in\ve{b}'} f^{-1}(b') \bigr) \cap \ve{a}$.
By construction, $\ve{b}\sqsubseteq\ve{a}$.
Also, $\ve{b}\in Q$, that is, it is a nonempty down-set, since
it is the intersection (meet) of two nonempty down-sets: 
$\ve{a}\in Q$, and
$\bigcup_{b'\in\ve{b}'} f^{-1}(b')\in Q$ by Lemma~\ref{lem:technical2},
and $Q$ is a lattice.
In addition, we need that $\ve{b}\in g^{-1}(\ve{b}')$, that is, $g(\ve{b}) = \ve{b}'$.
Since $g(\ve{a})=\ve{a}'$, we have that $\forall a'\in\ve{a}'$, $\exists a\in\ve{a}$ such that $f(a)=a'$.
Then also $\forall b'\in\ve{b}'\sqsubseteq'\ve{a}'$, $\exists a\in\ve{a}$ such that $f(a)=b'$,
that is, $a\in f^{-1}(b')$, and also $a\in\ve{a}$, so $a\in \ve{b}$ by construction of $\ve{b}$.\\
\eqref{eq:densu2} can be proven analogously.\\
\eqref{eq:sol2} is a simple consequence of the properties of $\sqsubseteq$.
Let $\ve{c}\sqsubseteq\ve{b}\sqsubseteq\ve{a}$,
then $g(\ve{c})\sqsubseteq'g(\ve{b})\sqsubseteq'g(\ve{a})$ by the monotonicity of $g$, 
so by the transitivity and antisymmetry of the partial order $\sqsubseteq'$, 
if $g(\ve{c}) = g(\ve{a})$ then $g(\ve{b})=g(\ve{a})$.\\
Note that the conditions \eqref{eq:constr}
were not used explicitly in these proofs,
they are assumed only for establishing the partial order $\preceq'$ for the set $P'$,
see Lemma~\ref{lem:fpo}.
\end{proof}

\begin{lem}
\label{lem:downup}
Lemmas~\ref{lem:technical1}, \ref{lem:technical2}, \ref{lem:com.set}, \ref{lem:com.rel} and \ref{lem:constrInherit},
hold also if up-sets are used instead of down-sets in the construction \eqref{eq:cd}.
\end{lem}

\begin{proof}
All the proofs can be repeated with slight modifications,
as interchanging the roles of \eqref{eq:densd} and \eqref{eq:densu},
down-closures and up-closures, and reversing orders.
\end{proof}

Thanks to 
Lemma~\ref{lem:com.set}, Lemma~\ref{lem:com.rel}, Lemma~\ref{lem:constrInherit} and Lemma~\ref{lem:downup},
the construction \eqref{eq:cd} can be repeated arbitrary times, with either up- or down-sets,
if the conditions \eqref{eq:constr} hold for the first level.
Note that $\sqsubseteq$ and $\sqsubseteq'$ are partial orders by construction anyway.
The inheritance of \eqref{eq:constr}, shown in Lemma~\ref{lem:constrInherit} is needed not for that,
but also for, e.g., Lemma~\ref{lem:com.set}, when applied for the next level.

\subsection{Embedding}
\label{app:structPI.emb}

Now, with the definition \eqref{eq:buildup},
let the function $g':Q'\to Q$ be given as
\begin{equation}
\label{eq:gp}
g'(\ve{a}') := \vee g^{-1}(\ve{a}'),
\end{equation}
and its image $Q'':=g'(Q')\subseteq Q$.
We will show that this is a subposet, and
$(Q'',\sqsubseteq) \isom (Q',\sqsubseteq')$, and $g'$ is an embedding of $Q'$ into $Q$.
(We use the notation $\vee A := \bigvee_{a\in A} a$ for any subset $A$ of a lattice.)

The following two technical lemmas concerning the images $g'(\ve{a}')$
will be used several times later.

\begin{lem}
\label{lem:Qppelements}
In the above setting
(and assuming \eqref{eq:constr}),
for all $\ve{a}'\in Q'$, we have 
\begin{subequations}
\begin{align}
\label{eq:gp1}
g'(\ve{a}') &=\sset{a\in P}{f(a)\in\ve{a}'},\\
\label{eq:gp2}
g'(\ve{a}') &=\bigcup_{a'\in\ve{a}'} f^{-1}(a').
\end{align}
\end{subequations}
\end{lem}

\begin{proof}
\eqref{eq:gp1} is by definition
$g'(\ve{a}')\equiv\vee \sset{\ve{a}\in Q}{g(\ve{a})=\ve{a}'}
=\sset{a\in P}{f(a)\in\ve{a}'}$, and we need to prove both inclusions.
First, $\vee \sset{\ve{a}\in Q}{g(\ve{a})=\ve{a}'} \subseteq\sset{a\in P}{f(a)\in\ve{a}'}$
holds, since $g$ is the elementwise action of $f$.
Second, for $\vee \sset{\ve{a}\in Q}{g(\ve{a})=\ve{a}'} \supseteq\sset{a\in P}{f(a)\in\ve{a}'}$,
we need to see that for all $a\in P$ such that $f(a)\in\ve{a}'$
there exists an $\ve{a}\in Q$ for which $a\in\ve{a}$, while $g(\ve{a})=\ve{a}'$.
The element $\ve{a} := \bigcup_{a'\in\ve{a}'} f^{-1}(a')$ fulfills these criteria,
we have $\ve{a}' = g(\ve{a})$ by Lemma~\ref{lem:technical1},
and $\ve{a}\in Q$ by Lemma~\ref{lem:technical2};
while $a\in\ve{a}$, since $f(a)\in\ve{a}'$, 
so it appears in the union $\bigcup_{a'\in\ve{a}'} f^{-1}(a')=f^{-1}(f(a))\cup \bigcup_{a'\in\ve{a}', a'\neq f(a)} f^{-1}(a')$.\\
\eqref{eq:gp2} can be proven by noting that
for the left-hand side, we have by \eqref{eq:gp1} that
$g'(\ve{a}')=\bigsset{a\in P}{f(a)\in\ve{a}'} \equiv \sset{a\in P}{\exists a'\in \ve{a}' \dispt{s.t.} f(a)=a'}$,
which is the same as the right-hand side
$\bigcup_{a'\in\ve{a}'} \sset{a\in P}{f(a)=a'}$.
\end{proof}

\begin{lem}
\label{lem:ggp}
In the above setting
(and assuming \eqref{eq:constr}),
for all $\ve{a}'\in Q'$, we have
$g(g'(\ve{a}'))=\ve{a}'$.
\end{lem}

\begin{proof}
Using \eqref{eq:gp2}, we have that $g(g'(\ve{a}'))=g\bigl(\bigcup_{a'\in\ve{a}'} f^{-1}(a')\bigr)$,
then Lemma~\ref{lem:technical1} leads to the claim.
\end{proof}

The following lemma shows that $g':Q'\to Q''\subseteq Q$ is an embedding of $Q'$ into $Q$.

\begin{lem}
\label{lem:embed}
In the above setting
(and assuming \eqref{eq:constr}),
$(Q'',\sqsubseteq) \isom (Q',\sqsubseteq')$.
\end{lem}

\begin{proof}
We need to prove that for all $\ve{a}',\ve{b}'\in Q'$,
$\ve{b}'\sqsubseteq'\ve{a}' \inliff g'(\ve{b}')\sqsubseteq g'(\ve{a}')$,
then, from the antisymmetry of the partial order, we have that $g'$ is bijective.\\
To see the \emph{``if'' direction}, we have
$g'(\ve{b}')\sqsubseteq g'(\ve{a}')$ then 
$g(g'(\ve{b}'))\sqsubseteq' g(g'(\ve{a}'))$ by the monotonicity of $g$, similarly to \eqref{eq:mon}, then
$\ve{b}'\sqsubseteq'\ve{a}'$ by Lemma~\ref{lem:ggp}.\\
To see the \emph{``only if'' direction}, we have
$\ve{b}'\sqsubseteq'\ve{a}'$, then
$ \sset{b\in P}{f(b)\in\ve{b}'} \sqsubseteq \sset{a\in P}{f(a)\in\ve{a}'}$, then 
$g'(\ve{b}')\sqsubseteq g'(\ve{a}')$ by \eqref{eq:gp1}.
\end{proof}

In the following four lemmas, we provide some simple tools,
concerning principal ideals.

\begin{lem}
\label{lem:principal}
In the above setting
(and assuming \eqref{eq:constr}),
for all $a\in P$, we have
$g(\downset\set{a})=\downset\set{f(a)}$.
\end{lem}

\begin{proof}
For the left-hand side we have
$g(\downset\set{a})
= g(\sset{b\in P}{b\preceq a})
= \sset{b'\in P'}{\exists b\in P  \dispt{s.t.} f(b)=b', b\preceq a}
= \sset{b'\in P'}{\exists b\in f^{-1}(b')  \dispt{s.t.} b\preceq a}$,
while for the right-hand side we have
$\downset\set{f(a)} = \sset{b'\in P'}{b'\preceq'f(a)}$, by definitions.
So we need that for all $a\in P$ and $b'\in P'$, 
$(\exists b\in f^{-1}(b')  \dispt{s.t.} b\preceq a) \inliff b'\preceq'f(a)$.\\
To see the \emph{``if'' direction}, we have
that if $b'\preceq'f(a)$ then for $a\in f^{-1}(f(a))$ there exists $b\in f^{-1}(b')$ such that $b\preceq a$ by \eqref{eq:densd}.\\
To see the \emph{``only if'' direction}, we have
that if $\exists b\in f^{-1}(b')$ such that $b\preceq a$, then $f(b)\preceq'f(a)$ by the monotonicity \eqref{eq:mon} of $f$,
then $b'\preceq'f(a)$, since $f(b)=b'$ by the assumption $b\in f^{-1}(b')$.
\end{proof}

\begin{lem}
\label{lem:principalback}
In the above setting
(and assuming \eqref{eq:constr}),
for all $a'\in P'$, we have
$g'(\downset\set{a'})=\downset f^{-1}(a')$.
\end{lem}

\begin{proof}
For the left-hand side we have
$g'(\downset\set{a'})
=\sset{b\in P}{f(b)\in\downset\set{a'}}
=\sset{b\in P}{f(b)\preceq a'}$ by applying \eqref{eq:gp1},
while for the right-hand side we have
$\downset f^{-1}(a') 
= \downset\sset{a\in P}{f(a)=a'}
= \sset{b\in P}{\exists a\in P \dispt{s.t.} b\preceq a, f(a)=a'}
= \sset{b\in P}{\exists a\in  f^{-1}(a') \dispt{s.t.} b\preceq a}$, by definitions.
So we need that for all $a'\in P'$ and $b\in P$
$(\exists a\in f^{-1}(a')  \dispt{s.t.} b\preceq a) \inliff f(b)\preceq'a'$.\\
To see the \emph{``if'' direction}, we have
that if $f(b)\preceq'a'$ then for $b\in f^{-1}(f(b))$ there exists $a\in f^{-1}(a')$ such that $b\preceq a$ by \eqref{eq:densu}.\\
To see the \emph{``only if'' direction}, we have
that if $\exists a\in f^{-1}(a')$ such that $b\preceq a$, then $f(b)\preceq'f(a)$ by the monotonicity \eqref{eq:mon} of $f$,
then $f(b)\preceq'a'$, since $f(a)=a'$ by the assumption $a\in f^{-1}(a')$.
\end{proof}

\begin{lem}
\label{lem:mittomen}
In the above setting
(and assuming \eqref{eq:constr}),
for all $a',b'\in P'$, we have
$b'\preceq'a' \inliff g'( \downset \set{b'})\sqsubseteq g'(\downset \set{a'})$.
\end{lem}

\begin{proof}
First, we have 
$b'\preceq'a' \inliff \downset\set{b'}\sqsubseteq'\downset\set{a'}$ obviously.
Second, we have
$\downset\set{b'}\sqsubseteq'\downset\set{a'} \inliff g'(\downset\set{b'})\sqsubseteq g'(\downset\set{a'})$,
which is a special case of Lemma~\ref{lem:embed}, for principal ideals in $Q'$.
\end{proof}

\begin{lem}
\label{lem:Qppelements2}
In the above setting
(and assuming \eqref{eq:constr}),
for all $\ve{a}'\in Q'$, we have 
\begin{equation}
\label{eq:gp3}
g'(\ve{a}') =\bigvee_{a'\in\ve{a}'} \downset f^{-1}(a').
\end{equation}
\end{lem}

\begin{proof}
We have
$\bigvee_{a'\in\ve{a}'} \downset f^{-1}(a')
= \bigvee_{a'\in\ve{a}'} g'(\downset\set{a'}) 
= \bigvee_{a'\in\ve{a}'} \bigcup_{b'\in\downset\set{a'}} f^{-1}(b')
= \bigcup_{a'\in\ve{a}', b'\preceq a'} f^{-1}(b') 
= \bigcup_{a'\in\ve{a}'} f^{-1}(b')
= g'(\ve{a}')$,
by using Lemma~\ref{lem:principalback}, 
\eqref{eq:gp2}, 
exploiting that $\ve{a}'$ is a down-set \eqref{eq:dset},
then \eqref{eq:gp2} again, respectively.
\end{proof}

Although it will not be used,
the following lemma shows that
the whole construction in this section works also if up-sets are used instead of down-sets.

\begin{lem}
\label{lem:downup2}
Lemmas~\ref{lem:Qppelements}, \ref{lem:ggp}, \ref{lem:embed}, \ref{lem:principal}, \ref{lem:principalback}, \ref{lem:mittomen} and \ref{lem:Qppelements2} 
hold also if up-sets are used instead of down-sets in the construction \eqref{eq:cd}.
\end{lem}

\begin{proof}
All the proofs can be repeated with slight modifications.
\end{proof}

\subsection{Application to the set and integer partitions}
\label{app:structPI.part}

Now we turn to the case of the main text,
for which the machinery developed in the previous sections is applied.
That is, for the roles of $(P,\finereq)$ and $(P',\finereq')$ we have
the poset $(P_\text{I},\preceq)$ of set partitions \eqref{eq:PI} with the refinement relation \eqref{eq:poI},
and the poset $(\pinv{P}_\text{I},\preceq)$ of integer partitions \eqref{eq:PIp} with the refinement relation \eqref{eq:poIp}.

First, we have the $\sigma\in\DscGrp{S}(L)$ permutations of the elementary subsystems $L$,
acting naturally on the partitions $\xi=\set{X_1,X_2,\dots,X_{\abs{\xi}}}\in P_\text{I}$ 
as $\sigma(\xi)=\bigsset{\sset{\sigma(i)}{i\in X}}{X\in\xi}\in P_\text{I}$.
We will make use of the obvious observations that
all the permutations of a partition $\xi\in P_\text{I}$
are of the same type, $s(\sigma(\xi))=s(\xi)\in\pinv{P}_\text{I}$;
and all partitions of a given type $\pinv{\xi}\in\pinv{P}_\text{I}$
can be transformed into one another by suitable permutations,
that is, for all $\xi_1,\xi_2\in s^{-1}(\pinv{\xi})$,
there exists at least one $\sigma\in\DscGrp{S}(L)$ by which $\xi_2=\sigma(\xi_1)$.
Also, for all $\sigma\in\DscGrp{S}(L)$ permutations,
$\upsilon\preceq\xi$ if and only if $\sigma(\upsilon)\preceq\sigma(\xi)$.
Using these, we can show that \eqref{eq:constr}-like conditions hold for this case.

\begin{lem}
\label{lem:ref.constr}
In the setting of the main text,
the following conditions hold
\begin{subequations}
\label{eq:ref.constr}
\begin{align}
\label{eq:ref.densd}
\begin{split}
&\forall \pinv{\xi},\pinv{\upsilon}\in\pinv{P}_\text{I}, \dispt{if} \pinv{\upsilon}\preceq\pinv{\xi} \dispt{then} \\
&\qquad\forall \xi \in s^{-1}(\pinv{\xi}), \; \exists \upsilon \in s^{-1}(\pinv{\upsilon}), \dispt{s.t.}  \upsilon\preceq\xi;
\end{split} \\
\label{eq:ref.densu}
\begin{split}
&\forall \pinv{\xi},\pinv{\upsilon}\in\pinv{P}_\text{I}, \dispt{if} \pinv{\upsilon}\preceq\pinv{\xi} \dispt{then} \\
&\qquad\forall \upsilon \in s^{-1}(\pinv{\upsilon}), \; \exists \xi \in s^{-1}(\pinv{\xi}), \dispt{s.t.}  \upsilon\preceq\xi;
\end{split} \\
\label{eq:ref.sol}
\begin{split}
&\forall \xi,\upsilon,\zeta\in P_\text{I}, \dispt{if} \zeta\preceq\upsilon\preceq\xi \\
&\qquad\dispt{and} s(\zeta) = s(\xi), \dispt{then} s(\upsilon)=s(\xi).
\end{split}
\end{align}
\end{subequations}
\end{lem}

\begin{proof}
Recall that $\pinv{\upsilon}\preceq\pinv{\xi}$ by definition if 
there exist $\upsilon_1\in s^{-1}(\pinv{\upsilon})$ and $\xi_1\in s^{-1}(\pinv{\xi})$, such that $\upsilon_1\preceq\xi_1$,
see \eqref{eq:poIp}, in accordance with \eqref{eq:relp}.
Then, 
for all $\xi_2\in s^{-1}(\pinv{\xi})$ we can construct $\upsilon_2$ which is $\upsilon_2\preceq\xi_2$,
simply as $\upsilon_2:=\sigma(\upsilon_1)$ by any of the permutations implementing $\xi_2=\sigma(\xi_1)$,
leading to \eqref{eq:ref.densd}.\\ 
\eqref{eq:ref.densu} can be proven analogously.\\
\eqref{eq:ref.sol} can be proven by noting that
for all $\xi,\upsilon\in P$ 
we have that if $\upsilon\preceq\xi$ then $\abs{\upsilon}\leq\abs{\xi}$, by \eqref{eq:poI};
and $\abs{\xi}=\abs{s(\xi)}$, by \eqref{eq:sI}.
Then $\zeta\preceq\upsilon\preceq\xi$ and $s(\zeta) = s(\xi)$
lead to $\abs{\zeta}\leq\abs{\upsilon}\leq\abs{\xi}$ with $\abs{\zeta}=\abs{\xi}$, 
form which we have $\abs{\upsilon}=\abs{\xi}$.
Now we need that $\abs{\upsilon}=\abs{\xi}$ and $\upsilon\preceq\xi$ lead to $\upsilon=\xi$ (then $s(\upsilon)=s(\xi)$),
which holds, because \eqref{eq:poI} must be fulfilled with the same number of parts:
let us number them accordingly as
$Y_i\subseteq X_i$ for $\xi=\set{X_1,X_2,\dots,X_k}$ and $\upsilon=\set{Y_1,Y_2,\dots,Y_k}$,
with the constraints of being disjoint and $\bigcup_{i=1}^k X_i = \bigcup_{i=1}^k Y_i = L$,
which leads to $Y_i=X_i$.
\end{proof}

Recall that Lemma~\ref{lem:fpo} states that 
if a poset $(P,\preceq)$ is transformed by a function $f$ as given in \eqref{eq:relp},
then $(P',\preceq')=(f(P),f(\preceq))$ is a poset if the conditions \eqref{eq:constr} hold.
Now Lemma~\ref{lem:ref.constr} shows that
for the poset $(P_\text{I},\preceq)$, if transformed by $s$ as given in \eqref{eq:PIp} and \eqref{eq:poIp} in the main text,
the corresponding conditions \eqref{eq:ref.constr} hold,
so $(\pinv{P}_\text{I},\preceq) = (s(P),s(\preceq))$ is a poset. 

Recall also that
if we form the down-set lattices \eqref{eq:buildup}, then
Lemma~\ref{lem:com.set} and Lemma~\ref{lem:com.rel} together state that \eqref{eq:cd} commutes,
Lemma~\ref{lem:constrInherit} states that the conditions in \eqref{eq:constr} are inherited.
Then with Lemma~\ref{lem:downup}, we have that
in the main text,
the three-level construction is well-defined as follows.

\begin{cor}
\label{cor:commut}
The diagram \eqref{eq:cd2} in the main text commutes.
\end{cor}

With Lemma~\ref{lem:principalback} and \eqref{eq:gp3}, we also have that
for the first two levels of the construction in the main text,
the following identities concerning the principal ideals hold.

\begin{cor}
In the setting of the main text,
for all 
$\pinv{\xi}\in\pinv{P}_\text{I}$ and
$\pinv{\vs{\xi}}\in\pinv{P}_\text{II}$, we have
\begin{subequations}
\begin{align}
\label{eq:labelsPpIpr}
\vee s^{-1} (\downset\set{\pinv{\xi}}) &= \downset s^{-1}(\pinv{\xi}),\\
\label{eq:labelsPpII}
\vee s^{-1}(\pinv{\vs{\xi}}) &= \bigvee_{\pinv{\xi}\in\pinv{\vs{\xi}}}\downset s^{-1}(\pinv{\xi}). 
\end{align}
\end{subequations}
\end{cor}

With Lemma~\ref{lem:mittomen} and Lemma~\ref{lem:embed}, we also have that
for the first two levels of the construction in the main text,
the following order isomorphisms hold.

\begin{cor}
In the setting of the main text,
for all 
$\pinv{\xi},\pinv{\upsilon}\in\pinv{P}_\text{I}$ and
$\pinv{\vs{\xi}}, \pinv{\vs{\upsilon}} \in\pinv{P}_\text{II}$, we have
\begin{subequations}
\begin{align}
\label{eq:embedPpIPII}
\pinv{\upsilon}\finereq\pinv{\xi}
  &\dspiff \vee s^{-1}(\downset\set{\pinv{\upsilon}}) \preceq \vee s^{-1}(\downset\set{\pinv{\xi}}),\\
\label{eq:embedPpIIPII}
\pinv{\vs{\upsilon}}\finereq\pinv{\vs{\xi}}
  &\dspiff \vee s^{-1}(\pinv{\vs{\upsilon}}) \preceq \vee s^{-1}(\pinv{\vs{\xi}}). 
\end{align}
\end{subequations}
\end{cor}

\section{Miscellaneous proofs}
\label{app:misc}

\subsection{Monotonicity of correlation measures}
\label{app:misc.xiLOmon}

Here we show
the $\xi$-LO monotonicity of the $\xi$-correlation \eqref{eq:CI}, from which
the LO monotonicity of the $\vs{\xi}$-correlation \eqref{eq:CII} follows,
because the latter one is a minimum of some of the former ones.
A quantum channel (completely positive trace preserving map \cite{Petz-2008,Wilde-2013}) $\Phi$ is $\xi$-LO,
if it can be written as $\Phi=\bigotimes_{X\in\xi}\Phi_X$ with quantum channels
$\Phi_X:\mathcal{D}_X\to\mathcal{D}_X$ for all $X\in\xi$.
With this,
\begin{equation*}
\begin{split}
C_\xi\bigl(\Phi(\varrho)\bigr) 
&\equalsref{eq:CI}     \min_{\sigma \in \mathcal{D}_{\xi\text{-unc}}}       D\bigl(\Phi(\varrho)\big\Vert\sigma\bigr)\\
&\lequalsref{eq:free}  \min_{\sigma'\in \Phi(\mathcal{D}_{\xi\text{-unc}})} D\bigl(\Phi(\varrho)\big\Vert\sigma'\bigr)\\
&\equals               \min_{\sigma \in \mathcal{D}_{\xi\text{-unc}}}       D\bigl(\Phi(\varrho)\big\Vert\Phi(\sigma)\bigr)\\
&\lequalsref{eq:contr} \min_{\sigma \in \mathcal{D}_{\xi\text{-unc}}}       D\bigl(\varrho\big\Vert\sigma\bigr)\\
&\equalsref{eq:CI}     C_\xi(\varrho),
\end{split}
\end{equation*}
where the first inequality comes from 
\begin{equation}
\label{eq:free}
\Phi(\mathcal{D}_{\xi\text{-unc}})\subseteq\mathcal{D}_{\xi\text{-unc}},
\end{equation}
which holds for $\xi$-LOs;
and the second inequality comes from the monotonicity of the relative entropy \eqref{eq:D},
\begin{equation}
\label{eq:contr}
D\bigl(\Phi(\varrho)\big\Vert\Phi(\sigma)\bigr)\leq D(\varrho\Vert\sigma),
\end{equation}
which holds for all quantum channels \cite{Petz-2008,Wilde-2013}.

Note that the monotonicity shown here is a particular case of a much more general property 
of monotone distance based \emph{geometric measures} \cite{Bengtsson-2006} in resource theories \cite{Chitambar-2019}:
the monotonicity with respect to free maps, by which the set of free states is mapped onto itself.

\bibliography{dualPP}{}

\end{document}